\documentclass[a4paper,10pt]{article}
\usepackage{graphicx} % Required for inserting images
\usepackage[utf8]{inputenc}
\usepackage[T1]{fontenc}
\usepackage{comment}
\usepackage[T1]{fontenc}
\usepackage{subfigure}
\usepackage{float}
\usepackage{natbib}

\usepackage{a4wide}
\usepackage{indentfirst}
\usepackage{authblk}

\usepackage{amsmath,amssymb}
\usepackage{amsthm}
\usepackage{hyperref}

\usepackage[margin=1.2in]{geometry}
\usepackage{graphicx}
\usepackage{listings}
\usepackage[utf8]{inputenc}
\usepackage{appendix}
\numberwithin{equation}{section}
\newtheorem{teo}{Theorem}[section]
\newtheorem*{teo*}{Theorem}
\newtheorem*{prop*}{Proposition}
\newtheorem*{corol*}{Corollary}
\newtheorem{prop}[teo]{Proposition}

\newtheorem{corol}[teo]{Corollary}
\newtheorem{lema}[teo]{Lemma}
\newtheorem{defi}[teo]{Definition}
\newtheorem{hyp}{Hypothesis}

\theoremstyle{definition}

\newtheorem{ex}[teo]{Example}

\newtheorem{remark}[teo]{Remark}

\renewcommand{\d}{\,\mathrm{d}}

\newcommand{\Bac}{\operatorname{Bac}}
\newcommand{\Dom}{\operatorname{Dom}}

\newcommand{\D}{\mathbb{D}}
\newcommand{\E}{\mathbb{E}}
\newcommand{\F}{\mathcal{F}}
\renewcommand{\H}{\mathcal{H}}

\newcommand{\R}{\mathbb{R}}

\newcommand{\1}{\mathbf{1}}

\newcommand{\eps}{\varepsilon}

\title{SHORT-TIME BEHAVIOR OF THE AT-THE-MONEY IMPLIED VOLATILITY FOR THE JUMP-DIFFUSION STOCHASTIC VOLATILITY BACHELIER MODEL}

\author[1]{Elisa Alòs}
\author[2]{Òscar Burés}
\author[3]{Josep Vives} 

\affil[1]{Universitat Pompeu Fabra and Barcelona School of Economics, Department of Economics and Business, Ramón Trias Fargas 25-27, 08005, Barcelona, Spain, \vspace*{3pt}}
\affil[2]{ Universitat de Barcelona, Facultat de Matemàtiques i Informàtica.
Gran Via de les Corts Catalanes, 585, 08007 Barcelona, Spain \vspace*{3pt}}
\affil[3]{Institut de Matemàtiques de la Universitat de Barcelona and Departament de Matemàtica econòmica, fincancera i actuarial. \authorcr Diagonal 690--696, 08034 Barcelona, Spain}
\date{\today}

\begin{document}

\maketitle

\begin{abstract}
    In this paper we use Malliavin Calculus techniques in order to obtain expressions for the short-time behavior of the at-the-money implied volatility (ATM-IV) level and skew for a jump-diffusion stock price. The diffusion part is assumed to be the stochastic volatility Bachelier model and the jumps are modeled by a pure-jump Lévy process with drift so that the stock price is a martingale. Regarding the level, we show that the short-time behavior of the ATM-IV level is the same for all pure-jump Lévy processes and, regarding the skew, we give conditions on the law of the jumps for the skew to exist. We also give several numerical examples of stochastic volatilities and Lévy processes that confirm the theoretical results found in the paper.
\end{abstract}



\section{Introduction}
Historically, quantitative finance has relied heavily on log-normal models such as the celebrated Black-Scholes model and its numerous extensions, including rough and local volatility models. A key feature of these models is that they ensure asset prices remain strictly positive. However, in recent years, instances of asset prices crossing the zero threshold and attaining negative values have been observed. Notably, in April 2020, during the COVID-19 recession, crude oil futures temporarily attained negative values due to the low demand compared to the storage cost. In response, exchanges such as the Chicago Mercantile Exchange (CME) and the Intercontinental Exchange (ICE) adapted their modeling frameworks, transitioning from the Black-Scholes model to the Bachelier model for oil and natural gas options until August 2020.

The Bachelier model, though historically overshadowed by log-normal models, has long played an important role in fixed income markets, where interest rates can take negative values, something poorly captured by log-normal models. For further discussion on the advantages and growing adoption of the Bachelier model, we refer to \cite{choi2022black}.

Despite its increasing popularity, analytical results for the Bachelier model remain relatively scarce. Recently, \cite{alos2023implied} provided explicit formulas for the short-time behavior of At-The-Money (ATM) implied volatility under the Bachelier model with local and fractional volatility. Although this work fills a significant gap, more generalizations are needed to account for market phenomena that cannot be fully explained by stochastic volatility alone. As discussed in \cite{tankov2003financial}, incorporating jump processes - specifically Lévy processes - into asset price dynamics allows modeling discontinuities and heavy-tailed behaviors, which are prevalent in financial markets. Moreover, combining stochastic volatility with jump-diffusion processes improves the fit of implied volatility surfaces.

In the Black-Scholes framework, short-time asymptotics for the ATM implied volatility level and skew under stochastic volatility and jump-diffusion models have been extensively studied. For example, \cite{alos2007short} examined the short-time behavior of the ATM implied volatility level and skew in the Bates model (see \cite{bates1996jumps}), while \cite{ROSINSKI2007677} analyzed similar properties when jumps follow a CGMY distribution.

In the present paper, we extend these analyses to the Bachelier setting by considering a jump-diffusion stochastic volatility model, where jumps follow a pure-jump Lévy process with drift. Using techniques from Malliavin calculus and classical results on Lévy processes, we establish that, for any pure-jump Lévy martingale, the ATM implied volatility level coincides with that of the pure stochastic volatility model without jumps (see \cite{alos2023implied} for a related result). Furthermore, if the jump process satisfies a stability condition, we derive an explicit formula for the short-time behavior of the implied volatility skew of the ATM, showing that it behaves analogously to the results obtained in \cite{alos2007short}, but in a normal setting rather than a log-normal one.

To validate our theoretical results, we performed numerical simulations under various stochastic volatility models, such as the SABR (see \cite{hagan2002managing}) and Rough Bergomi (see \cite{bayer2016pricing}) and different pure-jump Lévy processes. Specifically, we analyze cases where jumps follow a compound Poisson process with Gaussian and Laplace jumps (extending the Bates model to normal settings and more general volatility dynamics), as well as cases where jumps follow CGMY and Normal Inverse Gaussian (NIG) distributions.

The paper is structured as follows. In Section \ref{sec: Statement of the model and preliminaries} we define the model we assume for the stock price dynamics. In Section \ref{sec: Main result} we state the main result of the paper, Theorem \ref{th: main theorem} and the rest of the sections are devoted to prove the theorem and numerically check that the result holds. In fact, in Section \ref{sec: Hull-White formula}, we derive a Hull-White type formula that facilitates the proof of Theorem \ref{th: main theorem} in the compound Poisson case. In Section \ref{sec: Short-time behavior of the ATM Implied Volatility level} we prove Equation \eqref{eq: main theorem - level} for compound Poisson processes, that is, the part of Theorem \ref{th: main theorem} regarding the level  before extending it to general pure-jump Lévy processes. Section \ref{sec: Short-time behavior of the ATM Implied Volatility skew} follows a methodology similar to \cite{alos2007short} to prove equations \eqref{eq: main theorem - fractional case} and \eqref{eq: main theorem - rough case} of Theorem \ref{th: main theorem} for the compound Poisson case. In the same section, following an approximation argument, we prove Equations \eqref{eq: main theorem - fractional case} and \eqref{eq: main theorem - rough case} for the general Lévy case. The reason why we use an approximation argument is because the strategy of replicating the arguments in \cite{alos2007short} for a general pure-jump Lévy martingale is successful for Lévy processes with infinite activity and finite variation paths, but it fails for Lévy processes with infinite variation paths since the expansions do not give enough information about the leading term. Indeed, the leading order terms of the expansion in the infinite variation case involve a term that depends on an integral of a Greek of the Bachelier function with respect to the Lévy measure whose computation is, in general, not feasible.  Finally, in Section \ref{section: examples and numerical analysis}, we conduct numerical simulations to verify Theorem \ref{th: main theorem} across various well-established models in the literature.

\section{Statement of the model and preliminaries} \label{sec: Statement of the model and preliminaries}
Throughout the paper, we will assume zero interest rate $r = 0$. Not only we assume $r = 0$ for the sake of simplicity, but also because $r$ is assumed to be zero in interest rate models, while the underlying asset in commodity derivatives are future prices. Let $T > 0$ be a maturity time and let $S = \{S_t; t \in [0,T]\}$ follow the model
\begin{equation} \label{model of study}
S_t = S_0 + \int_0^t \sigma_s (\rho \d W_s + \sqrt{1-\rho^2} \d B_s) + L_t,
\end{equation}
where $S_0 > 0$ is fixed and $W$ and $B$ are independent Brownian motions on $[0,T]$ defined in the same probability space $(\Omega, \mathcal{F}, \mathbb{F},  \mathbb{P})$, where $\mathbb{F}$ is the filtration generated by $W$, $B$ and $L$. More precisely, let
\[
\mathbb{F}^{W} = \{\mathcal{F}^W_t; t \in [0,T]\}, \quad \mathbb{F}^{B} = \{\mathcal{F}^B_t; t \in [0,T]\} \quad \text{and } \quad\mathbb{F}^{L} = \{\mathcal{F}^L_t; t \in [0,T]\}
\]
denote the complete natural filtrations with respect to $W$,$B$ and $L$ respectively. Moreover, we define $\mathbb{F} = \{\mathcal{F}_t; t \in [0,T]\}$ as the filtration defined by $\mathcal{F}_t = \mathcal{F}^W_t \vee \mathcal{F}^B_t \vee \mathcal{F}^L_t$. We set $\rho \in (-1,1)$, we assume $\sigma$ to be an a.s. continuous, square integrable and $\mathbb{F}^W$-adapted stochastic process and $L$ denotes a pure-jump Lévy process with triplet $(a, 0, \nu)$. Because of the Lévy-Itô decomposition (see \cite{sato2013levy} or \cite{tankov2003financial}), $L$ is of the form
\[
L_t = a t + \int_0^t \int_{|y| \geq 1} y N(\d s, \d y) + \int_0^t \int_{|y| < 1} y \Tilde{N}(\d s, \d y), \quad t \in [0,T]
\]
where $N$ is a Poisson random measure and $\Tilde{N}(\d s,\d y) = N(\d s, \d y) - \d s \nu (\d y)$ denotes the compensated Poisson measure. Notice that, if we demand $S$ to be a martingale we need the condition
\[
a = -\int_{|y| \geq 1} y \nu(\d y).
\]
We also define the constants
\[
c_0 = \nu(\R) \in [0, \infty], \quad c_1 = \int_{\R} y \nu(\d y) \in [-\infty, \infty].
\]
We can classify $L$ depending on the finiteness of certain constants.
\begin{itemize}
    \item In the case $ c_0 = \nu(\R) < \infty$, $L$ is a compound Poisson process of intensity $\lambda = c_0$. Notice that if we let $\sigma$ follow the same dynamics as the volatility of the Heston model, then the model \eqref{model of study} becomes the Bachelier-Bates model.
    \item In the case $\nu(\R) = \infty$ but $\int_{\R} |y| \nu (\d y) < \infty$ then $L$ has infinite activity with finite variation trajectories, meaning that $L$ exhibits an infinitely amount of jumps but the paths of $L$ remain of bounded variation.
    \item In the case $\nu(\R) = \infty$ and $\int_{\R} |y| \nu(\d y) = \infty$ then $L$ has infinite activity with infinite variation trajectories. In that case, $c_1$ is not well-defined. However, as a consequence of the Lévy-Itô decomposition (see for instance \cite{sato2013levy} or \cite{tankov2003financial}), the process
    \[
    L_t^{\eps} = at +\int_{|y| \geq 1}y N(\d s, \d y) + \int_{\eps < |y| < 1} y \Tilde{N}(\d s, \d y)
    \]
    is a compound Poisson process with drift $a$ and intensity $c_0^{\eps} = \nu_{\eps}(\R) := \nu(\R \backslash [-\eps, \eps])$. In this situation, we can define the sequence of constants
    \[
    c_1^{\eps} = \int_{|y| > \eps} y\nu(\d y), \quad \eps > 0.
    \]
    Note that in the compound Poisson case, the infinite activity with finite variation case and some infinite activity and infinite variation cases, the sequence of constants $c_1^{\eps}$ converge to a real number $c_1$ as $\eps$ tends to zero.
\end{itemize}
Since $L_t$ is a pure-jump Lévy martingale, we can also write
\[
L_t = \int_0^t \int_{\R} y \tilde{N}(\d s, \d y), \quad L_t^{\eps} = \int_0^t \int_{|y| > \eps} y \tilde{N}(\d s, \d y).
\]
By abuse of language, we will refer to stochastic processes with infinite activity with (in)finite variation paths as \textit{infinite activity and (in)finite variation Lévy processes} or \textit{Lévy processes with infinite activity and (in)finite variation}. Regarding technical considerations on the volatility, first notice that when the volatility process is constant and there is no presence of jumps (i.e. $\sigma_t \equiv \sigma$ and $\nu \equiv 0$) then we recover the classical Bachelier model. 
\subsection{Malliavin Calculus tools}
In this section we will introduce the Malliavin calculus definitions and tools that are needed in order to obtain the key results of this paper. For a complete reference on the Malliavin calculus topics we refer the reader to \cite{nualart2006malliavin}. Let $W = \{W_t; t \in [0,T]\}$ be a standard Brownian motion defined in a complete probability space $(\Omega, \F, \mathbb{P})$. For a function $h \in L^2([0,T])$ we define $W(h)$ as the Wiener integral of $h$ with respect to $W$, i.e.
\[
W(h) := \int_0^T h(t) \d W_t.
\]
We define $\mathcal{S}$ as the class of random variables of $F$ the form
\[
F = f(W(h_1), \dots, W(h_n)); \quad h_i \in L^2([0,T]), \quad f\in \mathcal{C}^{\infty}_b(\R^n).
\]
We say that $\mathcal{S}$ is the space of \textit{simple functionals}. 
\begin{defi}
    Let $F \in \mathcal{S}$. We define the Malliavin derivative of $F$ evaluated at a point $t \in [0,T]$ as the stochastic process defined by
    \[
    D_t F = \sum_{j=1}^n \partial_j f(W(h_1), \dots, W(h_n)) h_j(t), \quad \partial_j f := \frac{\partial f}{\partial x_j}.
    \]
\end{defi}
As it has been proved in \cite{nualart2006malliavin}, the operator $D$ is closable from $L^p(\Omega)$ to $L^p(\Omega \times [0,T])$ for all $p \geq 1$. We define $\D^{1,p}_{W}$ as the closure of $\mathcal{S}$ with respect to the norm
\[
||F||_{1,p} = \left[ \E(|F|^p) + \E\left(||DF||_{L^2([0,T])}^p\right)\right]^{1/p}.
\]
In the same way we have defined $D$, we can define the iterated derivatives $D^n$. We denote $\D^{n,p}_{W}$ the space of random variables that are $n$ times Malliavin differentiable and $F$, together with its Malliavin derivatives are $p$-integrable.

The adjoint of the Malliavin derivative is the so-called Skorohod integral or divergence operator $\delta$. It is an unbounded and closed operator from $L^2(\Omega \times [0,T])$ to $L^2(\Omega)$. We say that a stochastic process $u$ belongs to the domain of $\delta$, $\Dom (\delta)$ if there exists a unique element $\delta(u)$ satisfying the duality relationship
\[
\E[F \delta(u)] = \E\left[ \int_0^T D_t F \cdot u_t \d t\right], \quad \forall F \in \D^{1,2}.
\]
 The Skorohod integral $\delta$ is an extension of the Itô integral (see \cite{nualart2006malliavin}), in the sense that if $u$ is a square-integrable adapted process then $\delta(u)$ coincides with the Itô integral $\int_0^T u_t \d W_t$ and $\delta(u)$ is well defined for a larger class of processes (such as, for instance, square-integrable non-adapted processes). Therefore, all integrals of the form $\int_0^T u_t \d W_t$ will be understood in the Skorohod sense. In order to get some insight of which processes belong to $\Dom(\delta)$, we introduce the following family of spaces.
\begin{defi}
    We define $\mathbb{L}^{k,p}_{W}$ as the space of stochastic processes $u$ such that for almost every $t \in [0,T]$, $u_t \in \D^{k,p}_{W}$.
\end{defi}
For all $k\geq 1$, $p \geq 2$ the space $\mathbb{L}^{k,p}_{W}$ is included in $\Dom(\delta)$. More explicitly, the following inclusions hold.
\begin{itemize}
    \item $\mathbb{L}^{k,p}_{W} \subset \mathbb{L}^{k', p}_{W}$ if $k' \leq k$.
    \item $\mathbb{L}^{k, p}_{W} \subset \mathbb{L}^{k, q}_{W} $ if $p \geq q$.
    \item $\mathbb{L}^{k,p}_{W} \subset \Dom(\delta)$ if $k \geq 1$, $p \geq 2$.
\end{itemize}
Again, for a proof of these results we refer the reader to \cite{nualart2006malliavin}.

\section{Main result} \label{sec: Main result}
In this section we state the main result of this paper. This result describes the short-time behavior of the ATM-IV for the jump-diffusion stochastic volatility Bachelier model. In order to state the result, we need some generic technical hypotheses on the volatility process $\sigma$. The set of hypotheses that will be assumed along the paper are the following.

\begin{hyp} \label{hyp 1}
    There exist $0 < \alpha < \beta$ such that
    \[
    \alpha \leq \sigma_t \leq \beta
    \]
    for almost every $t \in [0,T]$.
\end{hyp}
\begin{hyp} \label{hyp 2}
    For $p \geq 2$, $\sigma \in \mathbb{L}^{2,p}_W$.
\end{hyp}
\begin{hyp} \label{hyp 3}
    There exists $H \in (0,1)$ and, for all $p \geq 1$ there exist constants $\gamma,  \delta > 0$ such that for almost every $0 \leq r \leq s\leq u \leq T$ and for almost every $t$,
    \[
    E_t\left( |D_r\sigma_u|^p\right)^{1/p} \leq \gamma(u-r)^{H-1/2}
    \]
    and
    \[
    E_t\left( |D_s D_r \sigma_u|^p\right)^{1/p} \leq \delta (u-r)^{H-1/2}(u-s)^{H-1/2},
    \]
    where $E_t(\cdot)$ denotes the conditional expectation with respect to $\mathcal{F}_t \in \mathbb{F}$.
\end{hyp}

\begin{remark}
    The examples we consider in Section  \ref{section: examples and numerical analysis} don't satisfy Hypothesis \ref{hyp 1}. However, a classical truncation argument performed in the exact same way as in \cite{alos2023implied} and \cite{alos2019estimating} shows that the main result of the paper (see Theorem \ref{th: main theorem} below) still holds for the considered examples.
\end{remark}

Due to the assumption of $r = 0$, we can compute the value of a European call option with strike $k$ and maturity $T$ at time $t$ as $V_t = E_t[(S_T - k)_{+}]$. In the case of the classical Bachelier model (i.e. $\sigma_t \equiv \sigma$ and no jumps) then $V_t$ can be computed as follows:
\begin{equation} \label{eq: bac definition}
V_t = \Bac(T,t,S_t,k, \sigma) = (S_t - k) \Phi(d(k, \sigma)) +  \phi(d(k, \sigma)) \sigma \sqrt{T-t}
\end{equation}
where $\Phi$ denotes the cumulative distribution function of a standard normal random variable, $\phi = \Phi'$ and
\[
d(k, \sigma ) = \frac{S_t - k}{\sigma\sqrt{T-t}}.
\]
It is well known that the Bachelier function, $\Bac$, satisfies the following PDE
\[
\mathcal{L}(\sigma) \Bac(T,t,x,k,\sigma) := \partial_t \Bac(T,t,x,k, \sigma) + \frac{1}{2}\sigma^2 \partial^2_{xx} \Bac(T,t,x,k, \sigma) = 0.
\]
In order to introduce some notation, we define the following objects:
\begin{itemize}
    \item The future average volatility $v_t = \sqrt{\frac{Y_t}{T-t}}$ with $Y_t = \int_t^T \sigma_s^2 \d s$.
    \item The centered Gaussian kernel, $p(x,\tau)$, with variance $\tau^2$. If $\tau^2 = 1$ we simply write $p(x)$.
\end{itemize}
From \eqref{eq: bac definition} we can derive that
\[
\partial_{\sigma} \Bac(T,t,x,k,\sigma) = \phi(d(k,\sigma)) \sqrt{T-t} > 0.
\]
Hence, the function $\Bac$ is invertible with respect to $\sigma$. This allows us to define the Bachelier implied volatility $I_t(k)$ as the unique volatility parameter one should put in the Bachelier function in order to recover the market option price $V_t$, that is,
\[
V_t = \Bac(T,t,S_t, k, I_t(k)), \quad I_t(k) = \Bac^{-1}(T,t,S_t, k, V_t).
\]
If $I_t(k)$ is differentiable at least once with respect to the strike $k$, then Taylor's formula allows us to approximate $I_t(k)$ as
\[
I_t(k) \approx I_t(k^*_t) + \partial_k I_t(k^*_t)(k-k^*_t), \quad k_t^* = S_t.
\]
Therefore, the terms $I_t(k_t^*)$ and $\partial_k I_t(k^*_t)$ are key in order to get a linear approximation of the implied volatility surface.

In the present work, we apply Malliavin calculus techniques in order to understand the asymptotic behavior of these two terms as the maturity tends to zero depending on the regularity of the jumps. The main result of this paper is encapsulated in the following theorem.
\begin{teo} \label{th: main theorem}
    Assume $S_t$ follows the model \eqref{model of study} with $\sigma$ satisfying Hypotheses \ref{hyp 1}, \ref{hyp 2} and \ref{hyp 3} and let $L$ an arbitrary pure-jump Lévy martingale. Let $c_1^{\eps} = \int_{|y| > \eps} y \nu(\d y)$. Then,
    \begin{equation} \label{eq: main theorem - level}
    \lim_{T \to t} I_t(k^*_t) = \sigma_t.
    \end{equation}
    If moreover, the sequence $c_1^{\eps} = \int_{|y| > \eps} y \nu(\d y)$ has a finite limit $c_1$ as $\eps$ tends to zero, then
    \begin{equation}  \label{eq: main theorem - fractional case}
    \lim_{T \to t} \partial_k I_t(k^*_t) = \frac{c_1}{\sigma_t} + \lim_{T \to t} \frac{\rho}{\sigma_t (T-t)^2}\int_t^T \int_s^T E_t[D_s\sigma_u] \d u \d s
    \end{equation}
    if $H \geq 1/2$ and
    \begin{equation} \label{eq: main theorem - rough case}
    \lim_{T \to t} (T-t)^{1/2-H}\partial_k I_t(k^*_t) = \lim_{T \to t} \frac{\rho}{\sigma_t (T-t)^{3/2+H}}\int_t^T \int_s^T E_t[D_s\sigma_u] \d u \d s
    \end{equation}
    if $H < 1/2$.
\end{teo}
\begin{remark}
    If $L$ is a pure-jump Lévy process with finite variation then the condition $\lim_{\eps \to 0} c_1^{\eps} = c_1$ is always satisfied since the fact that $\int_{\R} |y| \nu (\d y) < \infty$ ensures that the constant $c_1 = \int_{\R} y \nu(\d y)$ is well defined and therefore $\lim_{\eps \to 0} c_1^{\eps} = c_1$. In the infinite variation case, this is no longer true. We will see in future sections that there exist some examples of infinite activity and infinite variation pure-jump Lévy processes for which the constants $c_1^{\eps}$ diverge as $\eps \to 0$ and, performing numerical simulations, we will see this effect in plots of the implied volatility.
\end{remark}

\begin{remark}
    Observe the following facts about Theorem \ref{th: main theorem}.
    \begin{itemize}
        \item If we assume no jumps (i.e. $\nu \equiv 0$ and therefore $c_1 = 0$) the result coincides with the short-time behavior studied in \cite{alos2023implied}. 
        \item The limit obtained in Equations \eqref{eq: main theorem - fractional case} and \eqref{eq: main theorem - rough case} is analogous to the findings in  \cite{alos2007short}, \cite{fukasawa2011asymptotic} and \cite{alos2008hull} where in this last reference the authors assume that the volatility process is also a Lévy process. The difference between Theorem \ref{th: main theorem} and the main result in \cite{alos2007short} is the definition of the constant $c_1$. This difference comes from the fact that in \cite{alos2007short} the authors work with a log-normal model while in this article we deal with a non log-normal model.
    \end{itemize}
\end{remark}

A study of the short-time behavior of the ATM-IV level and skew when the Lévy process has infinite activity and infinite variation paths has been addressed, for instance, in \cite{figueroa2015short},  \cite{figueroa2016short} and \cite{gerhold2016small}, where the authors consider different families of pure-jump Lévy processes with stochastic volatility driven by a Brownian motion and provide an expansion of the ATM-IV considering higher order terms. In the specific case where the process studied in \cite{figueroa2016short} is a CGMY process (see \cite{ROSINSKI2007677} or Section \ref{section: examples and numerical analysis} of this paper), the ATM-IV skew expansion up to first order derived by the authors closely resembles Equation \eqref{eq: main theorem - fractional case}. Notably, the cited references primarily focus on log-normal models. In general, the literature about the short-time behavior of the Implied Volatility level and skew is not very extensive. In fact, the literature concerning the jump-diffusion stochastic volatility Bachelier model is even scarcer. In this paper we aim to fill the gap in the literature concerning the ATM-IV level and skew for a very important non log-normal model.

The steps in order to prove Theorem \ref{th: main theorem} are the following:
\begin{enumerate}
    \item First, we prove the first part of the theorem in the case where $L$ is a compound Poisson process with drift following the same lines as in \cite{alos2007short} and \cite{alos2023implied}. 
    \item In order to extend the result to the most general case, we use an approximation argument, using the fact that the truncated Lévy process is a compound Poisson process.
    \item Finally, we do the same discussion for the skew using similar arguments.
\end{enumerate}

\section{A Hull-White type formula for call options} \label{sec: Hull-White formula}
In order to analyze the asymptotic behavior of the ATM implied volatility level and skew, we will make use of a decomposition formula for the prices of European call options. The first thing we need to do in order to derive a decomposition formula is to provide an anticipating Itô formula for our model.

\begin{prop} \label{prop: Anticipating Ito Formula}
Assume $S$ follows the model \eqref{model of study} with $\sigma$ satisfying Hypotheses \ref{hyp 1}, \ref{hyp 2} and \ref{hyp 3}. Let $F: [0,T] \times \R^2 \to \R$ a  $\mathcal{C}^{1,2,2}([0,T] \times \R^2)$ class function such that there exists a positive constant $C$ such that, for all $t \in [0,T]$, $F$ and its partial derivatives evaluated in $(t, S_t, Y_t)$ are bounded by $C$. Then, it follows that
\begin{align*}
    F(t, S_t, Y_t) = &F(0, S_0, Y_0) + \int_0^t \partial_{r} F(r, S_r, Y_r) \d r \\
    + & \int_0^t \partial_x F(r, S_r, Y_r) \sigma_r (\rho \d W_r + \sqrt{1-\rho^2} \d B_r) \\
    - & \int_0^t \partial_{y} F(r, S_r, Y_r)\sigma_r^2 \d r + \rho \int_0^t \partial^2_{xy}F(r,S_r, Y_r) (D^{-}Y)_r \sigma_r \d r \\
    + &\frac{1}{2}\int_0^t \partial_{xx}^2 F(r, X_r, Y_r) \sigma_r^2 \d r \\
    + & \int_0^t \int_{\R} \Delta_{yy}F(r, S_{r-}, Y_r) \nu (\d y) \d r \\
    + & \int_0^t \int_{\R} \left( F(r, S_{r-} + y, Y_r) - F(r,S_{r-}, Y_r)  \right)\Tilde{N}(\d r, \d y),
\end{align*}
where
\[
\Delta_{yy}F(r,S_{r-}, Y_r) = F(r, S_{r-}+y, Y_r) - F(r,S_{r-}, Y_r) -y\partial_x F(r, S_{r-}, Y_r)
\]
and, for a process of the form $Y_t = \int_t^T a_s \d W_s$ where $a$ is adapted and belongs to $\mathbb{L}^{1,2}$, $(D^{-}Y)_t := \int_t^T D_t a_s \d W_s$. Thus, for the particular case $Y_t = \int_t^T \sigma^2_s \d s $ we have
\[
(D^{-}Y)_t = \int_t^T D_t \sigma_s^2 \d s.
\]
\end{prop}
\begin{remark}
    An equivalent expression of this anticipating Itô formula is
    \begin{align*}
    F(t, S_t, Y_t) = &F(0, S_0, Y_0) + \int_0^t \partial_{r} F(r, S_r, Y_r) \d r \\
    + & \int_0^t \partial_x F(r, S_r, Y_r) \sigma_r (\rho \d W_r + \sqrt{1-\rho^2} \d B_r) \\
    - & \int_0^t \partial_{y} F(r, S_r, Y_r)\sigma_r^2 \d r + \rho \int_0^t \partial^2_{xy}F(r,S_r, Y_r) \Lambda_r \d r \\
    + &\frac{1}{2}\int_0^t \partial_{xx}^2 F(r, X_r, Y_r) \sigma_r^2 \d r \\
    + & \int_0^t \int_{\R} \Delta_{yy}F(r, S_{r-}, Y_r) \nu (\d y) \d r \\
    + & \int_0^t \int_{\R} \left( F(r, S_{r-} + y, Y_r) - F(r,S_{r-}, Y_r)  \right)\Tilde{N}(\d r, \d y),
\end{align*}
where
\[
\Lambda_r = \left( \int_r^T D_r \sigma_u^2 \d u \right) \sigma_r.
\]
\end{remark}
\begin{remark}
    The stock price $S$ exhibits at most a countable number of jumps since $L$ does so as well. Therefore, the set of times $r\in [0,T]$ where $S_{r-} \neq S_r$ has Lebesgue measure zero, so
    \[
    \int_0^t\int_{\R}\Delta_{yy} F(r,S_{r-}, Y_r)\nu (\d y) \d r = \int_0^t\int_{\R}\Delta_{yy} F(r,S_{r}, Y_r)\nu (\d y) \d r.
    \]
\end{remark}
\begin{proof}
    Consider, for any $\eps > 0$, the process
    \[
    S_t^{\eps} = S_0 + \int_0^t \sigma_r (\rho \d W_r + \sqrt{1-\rho^2}\d B_r) + \int_0^t \int_{|x|> \eps}x\Tilde{N}(\d s, \d x).
    \]
    Notice that this process has a finite number of jumps and converges to $S_t$ in $L^2$ as $\eps$ goes to zero. Fix $\eps > 0$, we name $T_i^{\eps}$ the times where $S_t$ exhibits a jump with the convention of $T_0^{\eps} = 0$. In the interval $[T_i^{\eps}, T_{i+1}^{\eps})$, $S_t$ evolves according to its continuous part $S_t^{\eps, c}$, that is,
    \[
    \d S_t^{\eps, c} =  \sigma_t \left( \rho \d W_t + \sqrt{1-\rho^2} \d B_t \right) - \left( \int_{|y| > \eps} y\nu(\d y)  \right)\d t.
    \]
    Hence, by applying the anticipating Itô formula (see \cite{alos2023implied}, \cite{alos2021malliavin} or \cite{alos2007short}) in the interval $[T_i^{\eps}, T_{i+1}^{\eps})$ we get
    \begin{align*}
        &F(T_{i+1}^{\eps}-, S^{\eps}_{T_{i+1}^{\eps}-}, Y_{T_{i+1}^{\eps}}) - F(T^{\eps}_i, S^{\eps}_{T^{\eps}_i}, Y_{T^{\eps}_i}) \\
        = &\int_{T^{\eps}_i}^{T^{\eps}_{i+1}-} \partial_r F(r, S^{\eps}_r, Y_r) \d r + \int_{T^{\eps}_i}^{T^{\eps}_{i+1}-} \partial_x F(r,S^{\eps}_r, Y_r) \d S_r^{\eps, c} \\
        -&\int_{T^{\eps}_i}^{T^{\eps}_{i+1}-} \partial_y F(r, S^{\eps}_r, Y_r) \sigma^2_r \d r + \rho \int_{T^{\eps}_i}^{T^{\eps}_{i+1}-} \partial_y F(r, S^{\eps}_r, Y_r) \Lambda_r  \d r \\
        +&\frac{1}{2} \int_{T^{\eps}_i}^{T^{\eps}_{i+1}-} \partial^2_{xx} F(r, S^{\eps}_r, Y_r) \sigma_r^2 \d r.
    \end{align*}
    Now, if a jump occurs at time $T_{i+1}^{\eps}$, its contribution is
    \[
     F(T_{i+1}^{\eps}, S^{\eps}_{T_{i+1}^{\eps}}, Y_{T_{i+1}^{\eps}}) - F(T_{i+1}^{\eps}, S^{\eps}_{T_{i+1}^{\eps}-}, Y_{T_{i+1}^{\eps}}).
    \]
    Hence, we can write
    \begin{align*}
        F(t, S^{\eps}_t, Y_t) = & F(0, S_0, Y_0) + \int_0^t \partial_r F(r, S^{\eps}_r, Y_r) \d r \\
        + &\int_0^t \partial_x F(r,S^{\eps}_r, Y_r) \sigma_r (\rho \d W_r + \sqrt{1-\rho^2} \d B_r) \\
        -&\int_0^t \int_{|y| > \eps} \partial_x F(r,S_r^{\eps}, Y_r)y \nu (\d y) \d s \\
        - &\int_0^t \partial_y F(r,S_r^{\eps}, Y_r) \sigma^2_r \d r + \rho \int_0^t \partial^2_{xy} F(r,S_r^{\eps}, Y_r) \Lambda_r \d r \\
        + &\frac{1}{2}\int_0^t \partial^2_{xx} F(r,S_r^{\eps}, Y_r) \sigma_r^2 \d r \\
        + &\sum_{i \geq 1} \left[ F(T_{i}^{\eps}, S^{\eps}_{T_{i}^{\eps}}, Y_{T_{i}^{\eps}}) -F(T_{i}^{\eps}, S^{\eps}_{T_{i}^{\eps}-}, Y_{T_{i}^{\eps}}) \right].
    \end{align*}
    Now, we can write
    \begin{align*}
       &\sum_{i \geq 1} \left[ F(T_{i}^{\eps}, S^{\eps}_{T_{i}^{\eps}}, Y_{T_{i}^{\eps}}) -F(T_{i}^{\eps}, S^{\eps}_{T_{i}^{\eps}-}, Y_{T_{i}^{\eps}}) \right] \\
       = &\int_0^t \int_{|y| > \eps} \left( F(r,S_{r-}+y, Y_r) - F(r,S_{r-}, Y_r) \right) N(\d y, \d r). \\
    \end{align*}
    Taking into account that $\tilde{N}(\d s, \d y) = N(\d s, \d y) - \nu (\d y) \d s$ and taking $\eps \downarrow 0$ we conclude the desired result. The details are omited since the proof follows the same scheme as in \cite{jafari2013hull}.
\end{proof}
The following lemma will be useful in order to justify that the integrals in the upcoming results are well defined, and will also be useful in the future to have an idea of the asymptotic behavior of the terms appearing in the decomposition formula.
\begin{lema} \label{lemma: well defined integral}
    Assume $S$ follows the model \eqref{model of study} with $-1 < \rho < 1$. Let $0 \leq t \leq r \leq T$, let $h \in \R$ and let $\mathcal{G}_t := \mathcal{F}_t \vee \F_T^{W} \vee \F^{L}_T $. Then, there exists a constant $C = C(n, \rho)$ such that
    \[
    \left| E\left( \partial^n_x G(T,r, S_r+h,k, v_r) | \mathcal{G}_t\right)\right| \leq C \left( \int_t^T\sigma_s^2 \d s\right)^{-\frac{1}{2}(n+1)}.
    \]
    where $G(T,t, S_t,k, v_t) := \partial^{2}_{xx}\Bac(T,t, S_t, k, v_t)$.
\end{lema}
\begin{proof}
    Recall that
    \[
    \Bac(T,t,k,x,\sigma) = (x-k)\Phi(d(k,\sigma)) + \Phi'(d(k,\sigma))\sigma \sqrt{T-t},
    \]
    where
    \[
    d(k,\sigma) = \frac{x-k}{\sigma\sqrt{T-t}}.
    \]
    A direct computation shows that
    \begin{align*}
        &\partial_x \Bac(T,t,k,x, \sigma)\\
        = &\Phi(d(k, \sigma)) +(x-k)\Phi'(d(k,\sigma)) \partial_x d(k,\sigma) + \Phi''(d(k,\sigma))\partial_{x} d(k,\sigma) \sigma\sqrt{T-t} \\
        = & \Phi(d(k,\sigma)) + \frac{x-k}{\sigma \sqrt{T-t}}\Phi'(d(k,\sigma)) + \Phi''(d(k,\sigma)) \\
        = &\Phi(d(k,\sigma)) + d(k, \sigma)\Phi'(d(k,\sigma)) - d(k,\sigma)\Phi'(d(k,\sigma)) \\
        =  &\Phi(d(k,\sigma)).
    \end{align*}
    Hence, $$G(T,t,x,k,\sigma) = \frac{\Phi'(d(k,\sigma))}{\sigma \sqrt{T-t}}. $$ This implies that
    \[
    E\left( \partial^n_{x} G(T,r, S_r + h,k,v_r) | \mathcal{G}_t \right)= (-1)^n \partial^n_{k} E\left( p(S_r + h-k, v_r \sqrt{T-r}) | \mathcal{G}_t \right).
    \]
   Now, the conditional expectation of $S_r+h$ given $\mathcal{G}_t$ is a normal random variable with mean equal to
   \[
    S_t + \rho \int_t^r \rho \sigma_{\theta} \d W_{\theta} + (L_r - L_t) + h =: \phi + h
   \]
   and variance equal to
   \[
   (1-\rho^2) \int_t^r \sigma^2_{\theta} \d \theta.
   \]
   This implies that 
   \begin{align*}
       &E\left( p(S_r + h- k, v_r \sqrt{T-r}) | \mathcal{G}_t\right) \\
       = & \int_{\R} p(y+h-k, v_r \sqrt{T-r}) p\left(y-\phi-h, \sqrt{ (1-\rho^2) \int_t^r \sigma^2_{\theta} \d \theta}\right) \d y \\
       = & p\left( \phi - k + 2h, \sqrt{\int_r^T \sigma^2_{\theta} \d \theta + (1-\rho^2)\int_t^r \sigma^2_{\theta} \d \theta}\right) \\
       = &p\left( \phi - k + 2h, \sqrt{(1-\rho)^2\int_t^T \sigma^2_{\theta} \d \theta + \rho^2\int_r^T \sigma^2_{\theta} \d \theta}\right).
   \end{align*}
   Finally, taking into account that $-1 < \rho < 1$ we find that
   \begin{align*}
       & \left|\partial^n_{k}p\left( \phi - k + 2h, \sqrt{(1-\rho)^2\int_t^T \sigma^2_{\theta} \d \theta + \rho^2\int_r^T \sigma^2_{\theta} \d \theta}\right) \right| \\
       \leq & C\left((1-\rho)^2\int_t^T \sigma^2_{\theta} \d \theta + \rho^2\int_r^T \sigma^2_{\theta} \d \theta \right)^{-\frac{1}{2}(n+1)} \leq C\left(\int_t^T \sigma_{\theta}^2 \d \theta \right)^{-\frac{n+1}{2}}_{,}
   \end{align*}
   concluding the proof.
\end{proof}
\begin{remark}
    Notice that the result still holds if $h = h_r$ is a stochastic process $\mathbb{F}$-adapted such that $|h_r|\leq M$ uniformly in $r$. Indeed, in this situation, due to the smoothness of the functions $\partial^n_x G$ and the fact that they are bounded in $x$ we know that
    \begin{align*}
    |E\left(\partial_x^n G(T,r,S_{r}+h_r, k, v_r)|\mathcal{G}_t\right)| \leq &\sup_{h \in [-M,M]}  | E\left(\partial_x^n G(T,r,S_{r}+h, k, v_r)|\mathcal{G}_t\right)|
    \end{align*}
    and we can apply now Lemma \ref{lemma: well defined integral} to deduce that the result holds. This remark generalizes a bit the similar results found in \cite{alos2007short} and \cite{alos2021malliavin}, for instance.
\end{remark}
 We are now ready to state and proof the Hull-White formula for Model \eqref{model of study}.
\begin{teo}
    Assume $S$ follows model \eqref{model of study} with $\rho \in (-1,1)$ and $\sigma$ satisfying Hypotheses \eqref{hyp 1} -- \eqref{hyp 3}. Then the value of a call option at time $t$, $V_t$ can be expressed as
    \begin{align*}
    V_t =  &E_t\left( \Bac(T,t, S_t, k, v_t)\right) \\
       + &\frac{\rho}{2} E_t \left(\int_t^T H(T,r,S_r, k, v_r) \Lambda_r \d r \right) \\
       + & E_t \left( \int_t^T \int_{\R}\Delta_{yy} \Bac(T,r,S_{r-}, k, v_r) \nu (\d y) \d r\right).
    \end{align*}
    where $G = \partial^2_{xx}\Bac$, $H = \partial_x G$ and $\Lambda_r = \sigma_r \int_r^T D_r \sigma_{\theta}^2 \d \theta$.
\end{teo}
\begin{proof}
    We use the anticipating Itô formula to $\Bac(T,T, S_T, k, v_T)$ to get
    \begin{align*}
        \Bac(T,T, S_T, k, v_T) = &\Bac(T,t, S_t, k, v_t) \\
        + &\int_t^T \partial_r \Bac(T,r,S_r, k, v_r) \d r \\
        + & \int_t^T \partial_{\sigma} \Bac(T,r,S_r, k, v_r) \frac{v_r^2-\sigma_r^2}{2v_r(T-r)} \d r \\
        + &\rho \int_t^T \partial^{2}_{x\sigma} \Bac(T,r,S_r, k, v_r) \frac{\Lambda_r}{2v_r(T-r)} \d r \\
        + &\frac{1}{2}\int_t^T \partial^2_{xx}\Bac(T,r,S_r, v_r) \sigma_r^2 \d r\\
        + &\int_t^T \int_{\R} \Delta_{yy} \Bac(T,r,S_r, k, v_r) \nu (\d y) \d r \\
        + & \int_t^T \partial_x \Bac(T,r, S_r, v_r) \sigma_r \left(\rho \d W_r + \sqrt{1-\rho^2} \d B_r \right).
    \end{align*}
   We now want to arrange the terms. Using the fact that
   \[
   \left(\partial_t +\frac{\sigma^2}{2}\partial^2_{xx} \right)\Bac(T,t,x,k,\sigma) = 0,
   \]
   and
   \[
   \partial^2_{xx}\Bac(T,t,x,k,\sigma) = \frac{\partial_{\sigma}\Bac(T,t,x,k,\sigma)}{\sigma(T-t)}
   \]
   we get
   \begin{align*}
       \Bac(T,T,S_T, k, v_T) = & \Bac(T,t, S_t, k, v_t) \\
       + &\frac{\rho}{2}\int_t^T \partial^3_{xxx}\Bac(T,r,S_r, k, v_r) \Lambda_r \d r \\
       + & \int_t^T \int_{\R} \Delta_{yy} \Bac(T,r,S_r, k, v_r) \nu (\d y) \d s \\
       + &  \int_t^T \partial_x \Bac(T,r, S_r, v_r) \sigma_r \left(\rho \d W_r + \sqrt{1-\rho^2} \d B_r \right).
   \end{align*}
   Taking conditional expectations to both sides and using that the stochastic integrals have zero expectation we conclude that
    \begin{align*}
       V_t =  &E_t\left( \Bac(T,T, S_T, k, v_T)\right) \\
       = &E_t\left( \Bac(T,t, S_t, k, v_t)\right) \\
       + &\frac{\rho}{2} E_t \left(\int_t^T H(T,r,S_r, k, v_r) \Lambda_r \d r \right) \\
       + & E_t \left( \int_t^T \int_{\R}\Delta_{yy} \Bac(T,r,S_{r-}, k, v_r) \nu (\d y) \d r\right).
    \end{align*}
   The second term in the right-hand-side of the previous formula is finite thanks to Lemma \ref{lemma: well defined integral}. In order to justify that the third term is well defined, we will follow closely the ideas of \cite{jafari2013hull}. We split the integral in the sum of the integrals over the sets $|y| \geq 1$ and $|y| < 1$ respectively, arguing in each case that the integrals are well defined. Regarding the integral when $|y| \geq 1$, we have
   \begin{align*}
     E_t\left( \int_t^T \int_{|y|\geq 1} \Bac(T,r,S_{r-}+y,k,v_r) - \Bac(T,r,S_{r-}, k, v_r) -y\partial_x \Bac(T,r,S_{r-}, k, v_r) \nu (\d y) \right).
   \end{align*}
   On the one hand we have $|\partial_x \Bac(T,t,x,k,\sigma)|\leq 1$, so $|\partial_x \Bac(T,r,S_r, k, v_r) y| \leq |y|$. Using the mean value theorem on the first two terms, we see that there exists a random point $\xi$ with $|\xi|\leq y$ such that
   \begin{align*}
     &E_t\left( \int_t^T \int_{|y|\geq 1} \Bac(T,r,S_{r-}+y,k,v_r) - \Bac(T,r,S_{r-}, k, v_r) -y\partial_x \Bac(T,r,S_{r-}, k, v_r) \nu (\d y) \right)\\
     \leq & E_t\left( \int_t^T \int_{|y|\geq 1} [\partial_x \Bac(T,r,S_{r-} + \xi, k, v_r) - \partial_x \Bac(T,r,S_{r-}, k, v_r)]y \nu(\d y) \d r\right).
   \end{align*}
    Since 
    \[
    \int_{|y|\geq 1} |y| \nu (\d y) < \infty,
    \]
    we have that this term is well defined. Regarding the integral when $|y| < 1$ we proceed as with $|y|\geq 1$, but we apply the mean value theorem one more time to get
    \[
    E_t\left(\int_t^T \int_{|y| < 1} G(T,r,S_r + \xi', k, v_r)y^2 \nu(\d y) \d r \right).
    \]
    In order to justify that this term is well defined, we apply the tower property of the conditional expectation to get
    \begin{align*}
    E_t\left(\int_t^T \int_{|y| < 1} G(T,r,S_r + \xi', k, v_r)y^2 \nu(\d y) \d r \right) = &E_t\left(E\left(\int_t^T \int_{|y| < 1} G(T,r,S_r + \xi', k, v_r)y^2 \nu(\d y) \d r \bigg| \mathcal{G}_t \right)\right) \\
    = & E_t\left(\int_t^T \int_{|y| < 1} E\left(G(T,r,S_r + \xi', k, v_r)\big| \mathcal{G}_t \right)y^2 \nu(\d y) \d r \right)
    \end{align*}
    Finally, applying Lemma \ref{lemma: well defined integral} together with the fact that
    \[
    \int_{|y|< 1} y^2 \nu (\d y)
    \]
    we deduce that the term is well defined, concluding the proof.
\end{proof}

\section{Short-time behavior of the ATM Implied Volatility level} \label{sec: Short-time behavior of the ATM Implied Volatility level}
The objective of this section is to study the short-time behavior of the ATM implied volatility. In other words, if we set $k_t^* = S_t$ we want to study the asymptotic behavior of $ I_t(k^*)$ as $T \to t$. As it was outlined in the previous sections, we will first consider the case where $L$ is a compound Poisson process with the drift that makes $L$ a martingale and then we will argue the other Lévy cases via an approximation argument. Since the result for the compound Poisson process can't be derived directly for every $\rho$, we will first prove a result for the uncorrelated case and then we will use this result to prove the correlated case.

\subsection{The uncorrelated case}
Our first analysis is focused on the case where $S$ follows Equation \eqref{model of study} with $\rho = 0$. We denote by $I_t^{0}(k)$ the implied volatility in this case. In this subsection, we prove the following result concerning the asymptotic behavior of the ATM implied volatility level as $T \to t$.
\begin{prop} \label{prop: uncorrelated level}
        Let $S$ follow the model \eqref{model of study} with $\rho = 0$ and let $L$ be a compound Poisson process with the drift term that makes it a martingale. Then,
        \begin{equation} \label{eq: level uncorrelated}
        \lim_{T \to t} I_t^0(k^*_t) = \sigma_t.
        \end{equation}
\end{prop}
\begin{proof}
Recall that the ATM-IV is defined via the equivalent relationships
\[
V_t = \Bac(T,t,S_t, k_t^*, I^0_t(k_t^*)), \quad I_t^0(k_t^*) = \Bac^{-1}(T,t, S_t, k_t^*, V_t).
\]
The Hull-White formula derived in Section \ref{sec: Hull-White formula} applied for the uncorrelated case ($\rho = 0$) tells us that
\begin{align*}
    V_t= & E_t(\Bac(T,t,S_t, k^*_t, v_t)) + E_t \left( \int_t^T \int_{\R} \Delta_{yy} \Bac(T,r,S_{r-}, k, v_r) \nu(\d y) \d r\right)_{k = k^*_t} \\
    = &E_t[T_1] + E_t[T_2].
\end{align*}
Hence,
\begin{align*}
    &I_t^0(k^*_t) \\
    = &E_t\left( \Bac^{-1}(T,t,S_t, k^*_t, E_t[T_1 + T_2]) \right) \\
    = &E_t\left(\Bac^{-1}(T,t,S_t, k^*_t, E_t[T_1+T_2]) 
    - \Bac^{-1}(T,t,S_t, k^*_t, T_1 + T_2)
    + \Bac^{-1}(T,t,S_t, k^*_t, T_1 + T_2) \right).
\end{align*}
We will analyze first the last term. Observe that we can write
\begin{align*}
    &E_t \left( \Bac^{-1}(T,t,S_t,k^*_t, T_1+T_2)\right)\\
    = & E_t \left(  \Bac^{-1}(T,t,S_t,k^*_t, T_1+T_2) - \Bac^{-1}(T,t,S_t, k^*_t, T_1)\right) + \E(\Bac^{-1}(T,t,S_t, k^*_t, T_1)).
\end{align*}
Using the mean value theorem, we observe that
\begin{align*}
    E_t \left(  \Bac^{-1}(T,t,S_t,k^*_t, T_1+T_2) - \Bac^{-1}(T,t,S_t, k^*_t, T_1)\right) 
     = &E_t \left(\partial_{\sigma}\Bac^{-1}(T,t,S_t, k^*_t, \xi)T_2 \right) \\
      = &E_t\left( \frac{\sqrt{2\pi}}{\sqrt{T-t}} T_2\right).
\end{align*}
Now, in order to analyze $T_2$, we resort to the mean value theorem in order to claim that exists $\xi_r$ such that
\begin{align*}
T_2 = &\int_t^T \int_{\R} \Delta_{yy} \Bac(T,r,S_{r-}, k, v_r) \nu(\d y) \d r
\\= &\int_t^T \int_{\R} (\partial_x \Bac(T,r, \xi_r, k, v_r) - \partial_x \Bac(T,r,S_{r-}, k, v_r) ) y\nu( \d y) \d r.
\end{align*}
Using now the facts that $|\partial_x\Bac(T,t,x,k,\sigma)|\leq 1$ and $\int_{\R} |y|\nu(\d y) < \infty$ we deduce that
\[
|T_2| \leq 2(T-t) ||y||_{L^1(\nu)}
\]
and therefore,
\[
E_t\left( \frac{\sqrt{2\pi}}{\sqrt{T-t}}T_2\right) \to 0, \text{ as } T \to t.
\]
On the other hand,
\begin{align*}
   E_t(\Bac^{-1}(T,t,S_t,k^*_t,T_1)) = &E_t(\Bac^{-1}(T,t,S_t, k^*_t, \Bac(T,t, S_t, k^*_t, v_t)) \\
    =  &E_t(v_t) \\
    \to & \sigma_t, \text{ as } T \to t.
    \end{align*}
    Finally, using the same arguments as in \cite{alos2023implied} for the term
    \[
    E_t\left(\Bac^{-1}(T,t,S_t, k^*_t, E_t[T_1+T_2]) - \Bac^{-1}(T,t,S_t, k^*_t, T_1 + T_2) \right)
    \]
    we see that this term goes to zero as $T \to t$, concluding therefore that the identity \eqref{eq: level uncorrelated} holds.
\end{proof}

    \begin{remark}
        Regarding the properties of $L$, we have relied on the fact that $||y||_{L^1(\nu)} < \infty$. Since this hypothesis is also satisfied for the infinite activity and finite variation case, we could have assumed that $L$ has infinite activity and finite variation and the result would be the same.
    \end{remark}
    \subsubsection{Extending the uncorrelated case to a general pure-jump Lévy process}
    In order to extend the result to the general pure-jump Lévy case we use an approximation argument with some suitable approximators. Let $S$ follow the model \eqref{model of study} and consider $(a, 0, \nu)$ the Lévy triplet of $L$. We construct these approximators as follows. Consider $0 < \eps <1$. Let $\tilde{S}^{\eps}$ be the process resulting of considering the model \eqref{model of study} with Lévy process $L^{\eps}$ defined as the Lévy process with triplet $(a,0, \nu_{\eps})$, where 
    \[
    \nu_{\eps} (\d x) = \1_{|x| > \eps} \nu(\d x).
    \]
    Notice that, by construction, $L_t^{\eps}$ has a finite amount of jumps and the magnitude of the jumps is greater than $\eps$. Moreover, from the fact that $L_t^{\eps}$ converges to $L_t$ as $\eps \to 0$ uniformly on $t$, we can assume (by replacing $\eps$ by a subsequence if needed) that $|L_t^{\eps} - L_t| < \eps / 2$. We denote by $\tilde{S}^{\eps}$ the stock price process that follows the same dynamics as in Equation \eqref{model of study} with Lévy martingale $L^{\eps}$. The sequence of approximators we will use is
    \[
    S_t^{\eps} = \tilde{S}_t^{\eps} + \frac{3\eps}{2}.
    \]
    It is easy to see that $S_t^{\eps} \to S_t$ a.s. and in $L^2$ uniformly on $t$. Moreover,
    \[
    S_t^{\eps} - S_t = \frac{3\eps}{2} + L_t^{\eps} - L_t \geq \frac{3 \eps}{2} - \frac{\eps}{2} = \eps > 0.
    \]
      We now want to proof that if $S$ follows the model \eqref{model of study} with $\rho = 0$ and $L$ being an arbitrary Lévy process with measure $\nu$, then the following result holds.
       \begin{prop} \label{prop: general level uncorrelated}
        Let $S$ follow the model \eqref{model of study} with $\rho = 0$ and $L$ being a general pure-jump Lévy martingale. Then,
        \[
        \lim_{T \to t} I_t^0(k^*_t) = \sigma_t.
        \]
    \end{prop}
      In order to prove the result, we will use a lemma that will help to get to the conclusion via an approximation argument. Consider $S_t^{\eps}$ the sequence of approximators defined as above. The result that we will use in order to prove Proposition \ref{prop: general level uncorrelated} is the following lemma.
      \begin{lema} \label{lemma: uniform convergence uncorrelated}
          Let $I_t^0(k_t^*)$ be the ATM implied volatility for the price model driven by $S$ and let $I_t^{0,\eps}(k_t^{\eps,*})$ be the ATM implied volatility for the price model driven by $S^{\eps}$. Then,
          \[
          \lim_{\eps \to 0} I_t^{0,\eps}(k_t^{\eps,*}) = I_t^0(k_t^*)
          \]
          uniformly in $T-t$.
      \end{lema}

      \begin{proof}
          Recall that if $V_t^{\eps}$ and $V_t$ stand for the price of a call option with strikes $k_t^{*.\eps} = S_t^{\eps}$ and $k^*_t = S_t$ under the stock models $S^{\eps}$ and $S$ respectively then
        \[
        I_t^{0,\eps}(k^{*,\eps}_t) = \Bac^{-1}(T,t,S_t^{\eps},k^{*,\eps}_t, V_t^{\eps}), \quad I_t^{0}(k^*) = \Bac^{-1}(T,t, S_t,k^*_t, V_t).
        \]
        Thus, we can write
        \begin{align*}
           &I_t^{0,\eps}(k^{*,\eps}_t) -  I_t^{0}(k_t^*) \\
           = &I_t^{0,\eps}(k^{*,\eps}_t) - \Bac^{-1}(T,t,S_t^{\eps},k_t^*, V_t^{\eps}) +  \Bac^{-1}(T,t,S_t^{\eps},k_t^*, V_t^{\eps}) - I_t^{0}(k_t^*).
        \end{align*}
        First, notice that by the mean value theorem
        \[
        I_t^{0,\eps}(k_t^{*,\eps}) - \Bac^{-1}(T,t,S_t^{\eps},k_t^*, V_t^{\eps}) = \partial_k \Bac^{-1}(T,t,S_t^{\eps}, \xi, V_t^{\eps})(k_t^{*,\eps} - k_t^{*})
        \]
        where 
        \begin{equation} \label{mean value}
            S_t = k_t^* \leq \xi \leq k_t^{*,\eps} = S_t^{\eps}.
        \end{equation}
        On the one hand,
        \[
        \partial_k \Bac^{-1}(T,t,S_t^{\eps}, \xi, V_t^{\eps}) =  \frac{-1}{\Phi(d(\xi,\Bac^{-1}(T,t,S_t^{\eps},\xi, V_t^{\eps})))}.
        \]
        Then, since we got
        \[
        \Bac(T,t,x,k,*) : (0, \infty) \to (0,\infty)
        \]
        whenever $x > k$, we have that
        \[
        \Bac^{-1}(T,t,S_t^{\eps},\xi, V_t^{\eps}) > 0
        \]
        This, together with the fact that \eqref{mean value} implies $S_t^{\eps} - \xi \geq 0$ leads us to
        \[
        |  I_t^{0,\eps}(k_t^{*,\eps}) - \Bac^{-1}(T,t,S_t^{\eps},k_t^*, V_t^{\eps})| \leq \frac{1}{\inf_{z \geq 0} \Phi(z)}|k_t^{*,\eps} - k_t^*|.
        \]
      On the other hand, using that $k_t^{*,\eps} = S_t^{\eps}$ and $k_t^* = S_t$ we see that $|  I_t^{0,\eps}(k_t^{*,\eps}) -\Bac^{-1}(T,t,S_t^{\eps},k_t^*, V_t^{\eps})|$ tends to zero uniformly on $T-t$. We now have to deal with $\Bac^{-1}(T,t,S_t^{\eps},k_t^*, V_t^{\eps}) - I_t^{0}(k^*_t)$. We can write
        \begin{align*}
             &\Bac^{-1}(T,t,S_t^{\eps},k_t^*, V_t^{\eps}) - I_t^{0}(k^*)\\
             = & \Bac^{-1}(T,t,S_t^{\eps},k^*_t, V_t^{\eps}) -  \Bac^{-1}(T,t, S_t,k^*_t, V_t) \\
             = & \Bac^{-1}(T,t,S_t^{\eps},k^*_t, V_t^{\eps}) -  \Bac^{-1}(T,t, S_t,k^*_t, V_t) + \Bac^{-1}(T,t,S_t, k_t^*, V_t^{\eps}) -\Bac^{-1}(T,t,S_t, k_t^*, V_t^{\eps}) \\
             = & \left(\Bac^{-1}(T,t,S_t^{\eps},k^*_t, V_t^{\eps}) - \Bac^{-1}(T,t,S_t, k_t^*, V_t^{\eps}) \right) + \left(\Bac^{-1}(T,t,S_t, k_t^*, V_t^{\eps})  -  \Bac^{-1}(T,t, S_t,k^*_t, V_t)  \right) \\
             = & A_1 + A_2.
        \end{align*}
        Let's analyze first $A_1$. Using the mean value theorem on the third variable we rewrite the first term as
        \[
        A_1 = \partial_x \Bac^{-1}(T,t,\xi, k_t^*, V_t^{\eps})(S_t^{\eps} - S_t).
        \]
        In the same way as we discussed before, we deduce that $|A_1| \to 0$ as $\eps \to 0$ uniformly on $T-t$. Regarding $A_2$, by means of the mean value theorem we know that there exists a mid point $\xi_{\eps}$ between $V_t$ and $V_t^{\eps}$ such that
        \[
        A_2 =  \partial_{\sigma}\Bac^{-1}(T,t,S_t,k^*_t, \xi_{\eps})(V_t^{\eps} - V_t).
        \]
        Hence,
        \[
       A_2 = \frac{ \sqrt{2\pi}}{\sqrt{T-t}}(V_t^{\eps} - V_t).
        \]
        Observe that
        \begin{align*}
        |V_t - V_t^{\eps}| = |E_t[(S_T - S_t)_+] - E_t[(S_T^{\eps} - S_t^{\eps})_+] |.
        \end{align*}
        Since the function $(\cdot)_{+}$ is Lipschitz continuous with Lipschitz constant $1$, we can rewrite the last expression as
        \[
        |V_t - V_t^{\eps}| \leq E_t\left( |(S_T - S_t) - (S_T^{\eps} - S_t^{\eps}) | \right),
        \]
        from which we deduce that
        \begin{align*}
            |V_t - V_t^{\eps}| \leq &E_t\left( |(S_T - S_t) - (S_T^{\eps} - S_t^{\eps}) | \right)
            \leq E_t\left(\left|\int_t^T \int_{|x| < \eps} x \Tilde{N}(\d s, \d x) \right| \right).
        \end{align*}
        Now, we use the isometry property of the compensated Poisson integral and the fact that it is a martingale with independent increments we have
        \[
        E_t\left[\left|\int_t^T \int_{|x| < \eps} x \tilde{N}(ds, dx)\right| \right] \leq \sqrt{(T-t) \int_{|x| < \eps} x^2 \nu(\d x)}.
        \]
        This implies that $|A_2| \to 0$ as $\eps \to 0$ uniformly on $T-t$.
      \end{proof}
    Using this lemma we can prove the desired result for a general pure-jump Lévy process.
     \begin{proof}[Proof (Proposition \ref{prop: general level uncorrelated})]
        Since the jumps of the approximator $S^{\eps}$ are modeled with a compound Poisson process with drift, we can rely on Proposition \ref{prop: uncorrelated level} to say that for every $\eps > 0$,
        \[
        \lim_{T \to t} I^{0, \eps}_t(k_t^{*,\eps}) = \sigma_t.
        \]
        Moreover, Lemma \ref{lemma: uniform convergence uncorrelated} gives us that
        \[
        \lim_{\eps \to 0} I^{0, \eps}_t(k_t^{*,\eps}) = I_t^0(k_t^*)
        \]
        uniformly on $T-t$. Hence, by the Moore-Osgood theorem we can switch limits and conclude that
        \begin{align*}
        \lim_{T \to t} I_t^0(k^*) = &\lim_{T \to t} \lim_{\eps \to 0}  I_t^{0,\eps}(k^{*,\eps}_t) \\
        =& \lim_{\eps \to 0} \lim_{T \to t} I_t^{0,\eps}(k_t^{*,\eps})\\
        = & \lim_{\eps \to 0} \sigma_t \\
        = & \sigma_t.
        \end{align*}
    \end{proof}
    With this result, we can now attack the correlated case.
    \subsection{The correlated case}
    In the same way as we have proved the uncorrelated case, we can prove that the ATM implied volatility for the correlated satisfies the same asymptotic behavior. The result stating this conclusion is the following.
    
    \begin{prop} \label{prop: level arbitrary Lévy}
        Let $S$ follow the model \eqref{model of study} with $L$ an arbitrary pure-jump Lévy martingale. Then,
        \[
        \lim_{T \to t} I_t(k^*_t) = \sigma_t.
        \]
    \end{prop}
    \begin{proof}
         We use the same strategy as with the uncorrelated case. We will deduce the result for the case where $L_t = -c_1t + J_t$ with $J$ a compound Poisson process since the approximation argument performed in the correlated case is identical as with the correlated case and therefore we skip it for the sake of conciseness. The Hull-White formula tells us that
    \begin{align*}
    V_t = &E_t(\Bac(T,t,S_t,k^*_t,v_t)) + \frac{\rho}{2}E_t\left[ \int_t^T H(T,r,S_r, k^*_t, v_r) \Lambda_r \d r\right] \\
    + &E_t\left[ \int_t^T \int_{\R} \Delta_{yy} \Bac(T,r,S_{r-}, k_t^*,v_r) \nu(\d y) \d r\right].
    \end{align*}
    Hence, if we let 
    \[
    \Gamma_{t,s} := E_t(\Bac(T,t,S_t, k_t^*, v_t))) + \frac{\rho}{2}E_t\left[ \int_t^s H(T,r,S_r, k^*_t, v_r) \Lambda_r \d r\right]
    \] and \[
    A_{t,T}:= E_t\left[ \int_t^T \int_{\R} \Delta_{yy} \Bac(T,r,S_{r-}, k_t^*,v_r) \nu(\d y) \d r\right]
    \]
    then we have
    \begin{align*}
        I_t(k^*_t) = &\Bac^{-1}(T,t,S_t, k^*_t, V_t) \\
        = & E_t \left( \Bac^{-1}(T,t,S_t, k^*_t, \Gamma_{t,T} + A_{t,T})\right) \\
        =  & E_t \left( \Bac^{-1}(T,t,S_t, k^*_t, \Gamma_{t,T} + A_{t,T}) -\Bac^{-1}(T,t,S_t, k^*_t, \Gamma_{t,T}) + \Bac^{-1}(T,t,S_t, k^*_t, \Gamma_{t,T}) \right).
    \end{align*}
    Now, using the mean value theorem, we have that there exists $\xi$ with
    \begin{align*}
        &E_t \left( \Bac^{-1}(T,t,S_t, k^*_t, \Gamma_{t,T} + A_{t,T}) -\Bac^{-1}(T,t,S_t, k^*_t, \Gamma_{t,T})\right) \\
        & = E_t \left( \partial_{\sigma} \Bac^{-1}(T,t,S_t, S_t, \xi)A_{t,T}\right)
    \end{align*}
    which, using the same argument as with the uncorrelated case, tends to zero as $T$ tends to $t$. For the remaining term we have
    \begin{align*}
        &E_t \left(\Bac^{-1}(T,t,S_t, k^*_t, \Gamma_{t,T}) \right) \\
        =  &E_t \left(\Bac^{-1}(T,t,S_t, k^*_t, \Gamma_{t,T})  - \Bac^{-1}(T,t,S_t, k^*_t, \Gamma_{t,t}) + \Bac^{-1}(T,t,S_t, k^*_t, \Gamma_{t,t})\right) \\
        = &I_t^{0}(k^*_t) + E_t \left(\Bac^{-1}(T,t,S_t, k^*_t, \Gamma_{t,T})  - \Bac^{-1}(T,t,S_t, k^*_t, \Gamma_{t,t}) \right)
    \end{align*}
    because $\Gamma_{t,t}$ corresponds to the price of a call option in the uncorrelated case. Using the classical Itô formula, we obtain
    \begin{align*}
         &E_t \left(\Bac^{-1}(T,t,S_t, k^*_t, \Gamma_{t,T})  - \Bac^{-1}(T,t,S_t, k^*_t, \Gamma_{t,t})\right) \\
         = & E_t \left( \int_t^T \partial_{\sigma} \Bac^{-1}(T,t,S_t,k^*_t, \Gamma_{t,r}) H(T,r,S_r, k ^*_t, v_r) \Lambda_r \d r\right).
    \end{align*}
    Using the same argument as in \cite{alos2021malliavin} and \cite{alos2023implied} we can conclude that
    \[
    E_t \left( \int_t^T \partial_{\sigma} \Bac^{-1}(T,t,S_t,k^*_t, \Gamma_r) H(T,r,S_r, k ^*, v_r) \Lambda_r \d r\right) \to 0, \quad \text{ as } T \to t.
    \]
    This implies that 
    \[
    \lim_{T \to t} I_t(k_t^*) = \lim_{T \to t} I_t^{0}(k_t^*) = \sigma_t.
    \]
   As mentioned at the beginning of the proof, the extension to the general Lévy case is proved using the same approximation argument as with the uncorrelated case.
    \end{proof}

    Notice that Proposition \ref{prop: level arbitrary Lévy} indeed proves the first part of Theorem \ref{th: main theorem}.

\section{Short-Time behavior of the ATM Implied Volatility Skew} \label{sec: Short-time behavior of the ATM Implied Volatility skew}
The objective now is to give an expression for the derivative with respect to the strike of the implied volatility under the Bachelier model in the at-the-money scenario, that is, when the strike equals the price of the stock at time $t$ (i.e. $k^*_t = S_t$). The first technical result that we need in order to understand the behavior of the skew as $T \to t$ is the following.
\begin{prop} \label{p: pre-expansion}
    Assume $S$ follows the model \eqref{model of study} with $\sigma$ satisfying Hypotheses \ref{hyp 1}, \ref{hyp 2} and \ref{hyp 3}. Then,
    \[
    \partial_k I_t(k^*_t) = \frac{E_t\left(\int_t^T \partial_k F(T, r,S_r, k, v_r) \d r \right)_{|k = k_t^*}}{\partial_{\sigma}\Bac(T,t, S_t, k_t^*, I_t(k^*_t))}
    \]
    where
    \begin{align*}
        F(T, r, S_r, k, v_r) = & \frac{\rho}{2}H(T,r, S_r, k, v_r) \Lambda_r + \int_{\R} \Delta_{yy}\Bac(T,r,S_{r-},k,v_r) \nu (\d y).
    \end{align*}
\end{prop}
\begin{proof}
    Since $V_t = \Bac(T,t, S_t, k, I_t(k))$, we can take partial derivatives with respect to $k$ in both sides of the equation to get
    \begin{equation} \label{numeric}
    \partial_k V_t(k) = \partial_k \Bac(T,t, S_t, k, I_t(k)) + \partial_{\sigma}\Bac(T,t, S_t, k, I_t(k)) \partial_k I_t(k).
    \end{equation}
    Furthermore, thanks to the Hull and White formula, we can write
    \[
    V_t(k) = E_t(\Bac(T,t,S_t, k, v_t) + E_t\left(\int_t^T F(T, r,S_r, k, v_r) \d r \right)
    \]
    where
    \begin{align*}
        F(T,r, S_r, k, v_r) = & \frac{\rho}{2}H(T,r, S_r, k, v_r) \Lambda_r
        +  \int_{\R}\Delta_{yy}\Bac(T,r,S_{r-},k,v_r) \nu (\d y).
    \end{align*}
    This readily implies that
    \[
    \partial_k V_t(k) = E_t(\partial_k \Bac(T,t, S_t, k, v_t)) + E_t\left( \int_t^T \partial_k F(T, r,S_r, k, v_r) \d r\right).
    \]
    We have to check that the two terms in the right hand side are well defined. Notice that we have the following relationship:
    \begin{align*}
        &\partial_k \Bac(T,t, S_t, k, v_t) \\
        = &-\Phi(d(k, v_t)) + (S_t - k)\Phi'(d(k,v_t))\partial_k d(k,v_t) + \Phi''(d(k,v_t))\partial_k d(k, v_t) v_t\sqrt{T-t} \\
        = & - \Phi(d(k,v_t)) - \Phi'(d(k,v_t)) d(k, v_t) - \Phi''(d(k,v_t)) \\
        = &- \Phi(d(k,v_t)) - \Phi'(d(k,v_t)) d(k, v_t) + d(k,v_t) \Phi'(d(k,v_t)) \\
        = & - \Phi(d(k,v_t)) \\
        = & -\partial_x \Bac(T,t, S_t, k, v_t).
    \end{align*}
    Since we know that $E_t(\partial_x \Bac(T,t, S_t, k, v_t))$ is well defined, we deduce that $E_t(\partial_k \Bac(T,t, S_t, k, v_t))$ is also well defined. This relationship also implies that
    \[
    \partial_k F(T,r,S_r, k, v_r) = - \partial_x F(T,r,S_r, k, v_r)
    \]
    so $E_t(\int_t^T \partial_k F(T, r, S_r, k, v_r) \d r)$ is well defined if and only if $E_t(\int_t^T \partial_x F(T, r, S_r, k, v_r) \d r)$ is well defined. Differentiating $F$ with respect to $x$ we obtain
    \begin{align*}
        \partial_x F(T, r,S_r, k, v_r)  = &\frac{\rho}{2} \partial_x H(T,r, S_r, k, v_r) \Lambda_r +\int_{\R}\Delta_{yy} \partial_x\Bac(T,r,S_{r-},k,v_r) \nu (\d y).
    \end{align*}
    Using Lemma \ref{lemma: well defined integral} and similar arguments as in the proof of the Hull-White formula we can show that
    \[
    E_t\left( \int_t^T \partial_xH(T,r,S_r, k, v_r) \Lambda_r \d r\right)
    \]
    is well defined and finite a.s. We shall now use similar arguments as in the proof of the Hull and White formula to conclude that the last term is well defined. For the term concerning $\nu$, we split the integral when $|y| \geq 1$ and when $|y|< 1$. If $|y|\geq 1$ we can apply the mean value theorem to get
    \[
    \partial_x \Bac(T,r,S_{r-}+ y, k, v_r) -  \partial_x \Bac(T,r,S_{r-}, k, v_r) = G(T,r,S_{r-} + \alpha y, k, v_r)y \quad \text{for some } \alpha \in [0,1].
    \]
    Hence, applying the tower property of the conditional expectation and Lemma \ref{lemma: well defined integral} to
    \[
    E\left(G(T,r,S_{r-} + \alpha y, k, v_r | \mathcal{G}_t \right) \quad \text{and} \quad E\left(G(T,r,S_{r-}, k, v_r | \mathcal{G}_t \right)
    \]
    we obtain the estimate
    \begin{align*}
        &E_t\left(\int_t^T \int_{|y|\geq 1} \Big[\partial_x \Bac(T,r,S_{r-}+ y, k, v_r)\right. \\
        &-  \partial_x \Bac(T,r,S_{r-}, k, v_r) - y G(T,r,S_{r-}, k, v_r) \nu \Big](\d y) \d r \Bigg) \\
        \leq & 2C (T-t)\left(\int_t^T \sigma_s^2 \d s \right)^{-1/2}E_t\left(\int_{|y|\geq 1} y \nu(\d y) \right).
    \end{align*}
    Using now the Hypothesis $\eqref{hyp 3}$ on $\sigma$ we conclude that there exists a constant $C' > 0$ such that
    \[
    2C (T-t)\left(\int_t^T \sigma_s^2 \d s \right)^{-1/2}E_t\left(\int_{|y|\geq 1} y \nu(\d y) \right) \leq 2C'(T-t)^{1/2}\int_{|y|\geq 1} y \nu(\d y)
    \]
    which is well defined because $\int_{|y| \geq 1} |y| \nu (\d y) < \infty$. For the second term, we use the mean value theorem twice and Lemma \ref{lemma: well defined integral} to get that
    \begin{align*}
           &E_t\left(\int_t^T \int_{|y|< 1} \Big[\partial_x \Bac(T,r,S_{r-}+ y, k, v_r)\right. \\
        &-  \partial_x \Bac(T,r,S_{r-}, k, v_r) - y G(T,r,S_{r-}, k, v_r) \nu \Big](\d y) \d r \Bigg) \\
        \leq & 2CE_t\left(\int_{|y|< 1} y^2 \nu(\d y) \right).
    \end{align*}
  Since all terms are well defined we find that the volatility skew satisfies
    \begin{align*}
        \partial_k I_t(k) = &\frac{\partial_k V_t(k) - \partial_k \Bac(T,t, S_t, k, I_t(k))}{\partial_{\sigma}\Bac(T,t, S_t, k, I_t(k))} \\
        = &\frac{E_t(\partial_k \Bac(T,t, S_t, k, v_t)) + E_t\left(\int_t^T \partial_k F(r,S_r, k, v_r) \d r \right) - \partial_k \Bac(T,t, S_t, k, I_t(k))}{\partial_{\sigma}\Bac(T,t, S_t, k, I_t(k))}.
    \end{align*}
    Finally, using the fact that
\[
\partial_k\Bac(T,t,S_t, k^{*}_t, v_t) = -\frac{1}{2}
\]
we have that, the ATM-IV skew satisfies
\[
\partial_k I_t(k^{*}_t) = \frac{E_t\left(\int_t^T \partial_k F(T, r,S_r, k, v_r) \d r \right)_{|k = k_t^*}}{\partial_{\sigma}\Bac(T,t, S_t, k_t^*, I_t(k^*_t))}
\]
as desired.
\end{proof}
Thanks to this proposition, we will derive the asymptotic behavior of the skew in the case where the driving Lévy process is a compound Poisson process and we will be able to prove Equations \eqref{eq: main theorem - fractional case} and \eqref{eq: main theorem - rough case} for this simplified scenario. In order to conclude the result for general pure-jump Lévy processes we will use an approximation argument similar as the one performed for the ATM-IV level.

\subsection{The compound Poisson case}
Since the anticipating Itô formula and the Hull-White formula were made for a general Lévy process, making extra assumptions on $L_t$ will simplify some formulas. Let's look at the case where $L_t$ is a compound Poisson process, that is, $\int_{\R} \nu(\d y)  = \nu (\R) = \lambda < \infty$. Moreover, the constants
\[
c_0 = \int_{\R} \nu(\d y) \quad \text{and} \quad c_1 = \int_{\R} y \nu (\d y)
\]
are finite. Notice that, in the compound Poisson process case, the terms concerning the Lévy measure can be separated. Indeed,
\[
 E_t\left( \int_t^T \int_{\R} \Delta_{yy}\Bac(T,r,S_{r-}, k, v_r) \nu (\d y) \d r \right)
\]
can be rearranged to get
\[
-c_1 E_t \left( \int_t^T \partial_x \Bac(T,r,S_r,k,v_r) \d r\right) + E_t\left(\int_t^T \int_{\R} \Delta_y \Bac(T,r,S_{r-},k,v_r) \nu(\d y) \d r \right)
\]
where $\Delta_y \Bac(T,r,S_{r-},k,v_r) =\Bac(T,r,S_{r-} + y,k,v_r) - \Bac(T,r,S_{r-},k,v_r)$. Hence, by Proposition \ref{p: pre-expansion} we can write $\partial_kI_t(k_t^*)$ as
\[
\partial_kI_t(k_t^*) = \frac{ E_t\left( \int_t^T \partial_k F(T,r,S_r, k_t^*, v_r) \d r\right)}{\partial_{\sigma} \Bac(T,t,S_t, k_t^*, v_t)}
\]
and therefore, the relationship
\[
\partial_{\sigma} \Bac(T,t,S_t, k_t^*, v_t) \partial_kI_t(k_t^*) = E_t\left( \int_t^T \partial_k F(T,r,S_r, k_t^*, v_r) \d r\right)
\]
holds. Since we are in the compound Poisson case, we can rewrite $E_t\left( \int_t^T \partial_k F(T,r,S_r, k_t^*, v_r) \d r\right)$ as 
\begin{align*}
   E_t\left( \int_t^T\partial_k F(T,r,S_r, k_t^*, v_r)\d r \right)= -&\frac{\rho}{2}E_t\left( \int_t^T \partial_x H(T,r,S_r, k, v_r) \Lambda_r \right)_{|k = k_t^*}\\
   - & E_t\left( \int_t^T \int_{\R} \Delta_y \partial_x\Bac(T,r,S_{r-}, k, v_r) \nu (\d y) \d r\right)_{|k = k_t^*} \\
   + &c_1E_t\left( \int_t^T G(T,r,S_r, k, v_r)y  \d r\right)_{|k = k_t^*} \\
   = & -T_1 - T_2 + T_3.
\end{align*}
This allows us to deduce the following proposition.
\begin{prop}
        Consider $S$ following the model \eqref{model of study} with $L$ a compound Poisson process. Then, the skew satisfies 
        \begin{align*}
            \partial_{\sigma}\Bac(T,t,S_t,k_t^*, v_t) \partial_kI_t(k_t^*)= -&\frac{\rho}{2}E_t(\partial_xH(T,t,S_t,k_t^*, v_t) J_t)+ c_1E_t(G(T,t,S_t, k_t^*, v_t)(T-t)) \\
            + & O(T-t)^{\min(2H,1)}.
        \end{align*}
    \end{prop}
    \begin{proof}
    The proof of this result follows the same ideas as in \cite{alos2007short}. The strategy consists in analyzing the three terms $T_1$, $T_2$ and $T_3$.
    
        \textit{Step 1:} The idea to analyze $T_1$ is applying the anticipating Itô formula (Proposition \ref{prop: Anticipating Ito Formula}) to $\frac{\rho}{2}\partial_k H(T,t, S_t, k, v_t) \int_t^T \Lambda_s \d s$. In order to make the notation easier, we define $J_t = \int_t^T \Lambda_r \d r$. We will apply the Itô formula to the process $\mathcal{H}(t,S_t,k, v_t)J_t$ where $\H = \partial_x H$ and then we will deduce the result multiplying everything by $\rho/2$ and changing the sign. We have
    \begin{align*}
        &E_t\left(  \int_t^T \H(T,r,S_r, k, v_r) \Lambda_r \d r\right)\\
        = &E_t\left(\H(T,t,S_t, k, v_t) J_t \right)  \\
        +  &\frac{\rho}{2}E_t\left(\int_t^T \partial^3_{xxx}\H(T,r,S_r, k, v_r)J_r \Lambda_r \d r \right) \\
        + &\rho E_t\left(\int_t^T \partial_x \H(T,r,S_r, k, v_r) (D^{-}J)_r \sigma_r \d r \right) \\
        + &E_t\left(\int_t^T \int_{\R}\Delta_{y}\H(T,r,S_{r-},k,v_r) J_r \nu(\d y) \d r \right) \\
        - & c_1E_t\left( \int_t^T \partial_x \H(T,r,S_{r}, k, v_r) J_r \d r\right) \\
        = & E_t\left(\H(T,t,S_t, k, v_t) J_t \right) + \sum_{j=1}^4 A_j.
    \end{align*}
   Observe that
   \begin{align*}
   \int_t^T |\Lambda_r| \d r \leq &\int_t^T |\sigma_r|\int_r^T 2|D_r\sigma_u| \d u \d r\\
   \leq &C (T-t) \int_r^T |D_r\sigma_u| \d r \d u \\ \leq &C (T-t)\left( \int_t^T \int_r^T (D_r\sigma_u)^2 \d u \d r \right)^{1/2}.
   \end{align*}
   Moreover, using Hölder's inequality, we have that
   \begin{align*}
       |J_r| \leq &2\left( \int_r^T \sigma_{\theta}^2 \d \theta\right)\left(\int_r^T \int_u^T (D_u \sigma_{\theta})^2 \d \theta \d u \right)^{1/2} \\
       \leq & 2C(T-r)\left(\int_r^T \int_u^T (D_u \sigma_{\theta})^2 \d \theta \d u \right)^{1/2}.
   \end{align*}
   Now, using that $|E(\partial_{xxx}^3 \mathcal{H}(r,S_r, k, v_r) | \mathcal{G}_t)| \leq C(T-t)^{-3}$ we have that
   \begin{align*}
   |E_t\left(\int_t^T \partial^3_{xxx}\H(T,r,S_r, k, v_r)J_r \Lambda_r \d r \right)| \leq &\frac{C}{T-t}E_t\left(\int_t^T \int_r^T (D_r\sigma_{\theta})^2 \d \theta \d r \right) \\
   = &O(T-t)^{2H},
   \end{align*}
    so $A_1 = O(T-t)^{2H}$. Let's deal now with $A_2$. We have $(D^{-}J)_r = \int_r^{T} D_r\Lambda_u \d u$. This implies that
    \begin{align*}
        \int_r^T D_r\Lambda_u \d u = &\int_r^T D_r\left(\sigma_u \int_u^T D_u \sigma^2_{\theta} \d \theta\right) \d u\\
        = &\int_r^T \left[D_r\sigma_u \int_u^T D_u \sigma_{\theta}^2 \d \theta + \sigma_u \int_u^T D_r D_u \sigma^2_{\theta} \d \theta \right]\d u.
    \end{align*}
    Now,
    \[
    D_r D_u \sigma^2_{\theta} = 2D_r(\sigma_u D_u \sigma_{\theta}) = 2(D_r \sigma_u D_u \sigma_{\theta} + \sigma_u D_r D_u \sigma_{\theta}).
    \]
    Hence,
    \begin{align*}
       E_t\left( \sigma_u \int_u^T D_r D_u \sigma_{\theta}^2 \d \theta\right) \leq &C\int_u^T (\theta-r)^{H-1/2}(\theta-u)^{H-1/2} \d \theta \\
        \leq & (T-r)^{2H}
    \end{align*}
    and
    \begin{align*}
        E_t\left(D_r\sigma_u \int_u^T D_u \sigma_{\theta}^2 \d \theta \right) \leq &CE_t\left(D_r\sigma_u \int_u^T D_u \sigma_{\theta} \d \theta \right) \\
        \leq & C(u-r)^{H-1/2}\int_u^T (\theta-u)^{H-1/2} \d \theta \\
        = & C(u-r)^{H-1/2}(T-u)^{H+1/2}.
    \end{align*}
    This implies that
    \begin{align*}
        E_t\left( \int_r^T D_r\Lambda_u \d u\right) \leq& C\int_r^T (T-r)^{2H} + (u-r)^{H-1/2}(T-u)^{H+1/2} \d u \\
        = & C(T-r)^{2H+1} + \int_r^T (u-r)^{H-1/2}(T-u)^{H+1/2} \d u \\
        \leq & C (T-r)^{2H+1}.
    \end{align*}
    Finally, using again Lemma \ref{lemma: well defined integral} in $A_2$ we have
    \begin{align*}
        A_2 = &E_t\left( \int_t^T  E\left[\partial^3_{xxx}G(T,r,S_r, k, v_r) |\mathcal{G}_r \right] (D^{-}J)_r \sigma_r \d r\right) \\
        \leq & CE_t\left( \int_t^T \left(\int_t^T \sigma_r^2 \d r \right)^{-2} \left(\int_r^T D_r\Lambda_u \d u\right) \d r \right) \\
        \leq & (T-t)^{-2} \int_t^T C(T-r)^{2H+1} \d r \\
         = & O(T-t)^{2H}.
    \end{align*}
    For $A_3$, we use that $|E\left(\Delta_y \H(T,r,S_r, k, v_r) | \mathcal{G}_t \right)| \leq 2C(T-t)^{-3/2}$ and the fact that $\nu(\R) < \infty$ to conclude that
    \[
    A_3 \leq C(T-t)^{-3/2}E_t\left( \int_t^T J_r \d r \right) = O(T-t)^{H+1}.
    \]
    Since $H > 0$, this also implies that
    \[
    A_3 = O(T-t).
    \]
    Regarding $A_4$, we can use similar arguments to derive that
    \[
    A_4 \leq C(T-t)^{-2} E_t\left( \int_t^T J_r \d r \right) = O(T-t)^{H+1/2}.
    \]
    Now, since for $H > 1/2$ we have $H + 1/2 > 1$ we have
    \[
    A_4 = O(T-t)
    \]
    if $H > 1/2$. If $H < 1/2$ then $H + 1/2 > 2H$, so
    \[
   A_4 = O(T-t)^{2H}
    \]
    whenever $H < 1/2$. Therefore, $A_4 = O(T-t)^{\min(2H,1)}$. Plugging the asymptotics of the terms $A_j$, $j = 1,2,3,4$ together we conclude that 
    \[
    T_1 = E_t\left(\H(T,t,S_t, k, v_t) J_t \right) + O(T-t)^{\min(2H, 1)}.
    \]
    
    \textit{Step 2:} Using that $|\Delta_y \partial_x \Bac(T,r,S_{r-}, k, v_r)| \leq 2$ and $\nu(\R) < \infty$ we can easily see that $T_2 = O(T-t)$.
    
    \textit{Step 3:} Regarding $T_3$, we use the anticipating Itô formula to $G(T,r,S_r, k, v_r)$ to deduce that
    \begin{align*}
        G(T,r,S_r, k, v_r) = &G(T,t,S_t, k, v_t) + \frac{\rho}{2}E_t\left( \int_t^r \partial^3_{xxx} G(T,u, S_u, k, v_u) \Lambda_u \d u\right) \\
        - &c_1 E_t\left( \int_t^r \partial_x G(T,u,S_u, k, v_u) \d u\right) \\
        + & E_t\left( \int_t^r \int_{\R} \Delta_{y} G(T,u, S_{u-}, k, v_u) \nu(\d y) \d u\right).
    \end{align*}
    Now, using Lemma \ref{lemma: well defined integral} in all terms together with  the fact that $\int_{\R} y  \nu(\d y) < \infty$ and integrating from $t$ to $T$ with respect to $r$ we find that 
    \[
    T_3 = c_1 E_t(G(T,t,S_t,k_t^*, v_t)(T-t)) + O(T-t)^{\min(2H,1)}.
    \]
    Joining the estimates for $T_1$, $T_2$ and $T_3$ we obtain the desired result.
    \end{proof}
Notice that when we are at the ATM case (i.e. $k = k_t^* = S_t$) we can write the Greeks of the Bachelier function as follows:
\[
\partial_{\sigma} \Bac(T,t,S_t, k^*_t, I_t(k^*_t)) = \frac{\sqrt{2\pi}}{\sqrt{T-t}},
\]
\[
G(T,t,S_t,k^*_t, v_t) = \frac{1}{\sqrt{2\pi}} \frac{1}{v_t \sqrt{T-t}},
\]
and
\[
\partial_k H(T,t,S_t, k^*_t, v_t) = \frac{1}{\sqrt{2\pi}} \frac{1}{v_t^3 (T-t)^{3/2}}.
\]
These explicit expressions of the Greeks are very helpful in order to prove the asymptotic of the skew when the Lévy process is a compound Poisson process with drift.
\begin{prop} \label{prop: cpp}
    Let $S$ be a stock price with dynamics following model \eqref{model of study} with $L$ a compound Poisson process. Then,
    \begin{itemize}
        \item If $H \geq 1/2$, then
        \[
        \lim_{T \to t} \partial_kI_t(k^*_t) = \frac{c_1}{\sigma_t} + \lim_{T \to t} \frac{\rho}{\sigma_t (T-t)^2}\int_t^T \int_s^T E_t[D_s\sigma_u] \d u \d s.
        \]
        \item If $H < 1/2$, then
        \[
        \lim_{T \to t} (T-t)^{1/2-H}\partial_k I_t( k^*_t) = \lim_{T \to t} \frac{\rho}{\sigma_t (T-t)^{3/2+H}} \int_t^T \int_s^T E_t[ D_s \sigma_u] \d u \d s.
        \]
    \end{itemize}
\end{prop}
\begin{proof}
Using the previous computations of the Greeks, we can write the ATM implied volatility skew as
\begin{align*}
    \partial_kI_t(k^*_t) = &\frac{\rho}{2}E_t\left[ \frac{1}{v_t^3 (T-t)^2}\int_t^T \Lambda_s \d s\right] + c_1 E_t\left[\frac{1}{v_t}\right] + O(T-t)^{\min(2H,1)-1/2} \\
= & B_1 + B_2 + O(T-t)^{\min(2H,1)-1/2}.
\end{align*}
Let's first consider the case where $H \geq 1/2$ so that $\min(2H, 1) = 1$. On the one hand, since $\lim_{T \to t} v_t = \sigma_t$ and $\sigma_t$ is bounded from below uniformly in $t$ we have that, due to the dominated convergence theorem,
\[
\lim_{T \to t} c_1E_t\left[\frac{1}{v_t}\right] = \frac{c_1}{\sigma_t}.
\]
On the other hand
\[
B_1 = \frac{\rho}{ (T-t)^2} \E \left[ v_t^{-3}\int_t^T \sigma_s \int_s^T \sigma_u D_s\sigma_u \d u \d s\right].
\]
Then, using the dominated convergence theorem again, we find that
\[
\lim_{T \to t} B_1 = \lim_{T \to t} \frac{\rho}{\sigma_t (T-t)^2}\int_t^T \int_s^T E_t[D_s\sigma_u] \d u \d s.
\]
Hence,
\[
\lim_{T \to t} \partial_kI_t(k^*_t) = \frac{c_1}{\sigma_t} + \lim_{T \to t} \frac{\rho}{\sigma_t (T-t)^2} \int_t^T \int_s^T E_t[D_s\sigma_u] \d u \d s,
\]
where this last limit can be computed explicitly depending on the model. 

Let's consider now the case where $H < 1/2$. In this case, $O(T-t)^{\min(2H, 1)} = O(T-t)^{2H}$ we now aim to proceed in a similar way, but we have to take into account that $2H -1/2$ is not positive for all $H$. Indeed, if $H < 1/4$, then $2H < 1/2$. However, since $H > 0$, we can use the fact that $(1/2 - H) + (2H-1/2) > 0$ for all $H > 0$ to conclude that
\[
\lim_{T \to t} (T-t)^{1/2-H}\partial_k I_t( k^*_t) = \lim_{T \to t} \frac{\rho}{\sigma_t (T-t)^{3/2+H}} \int_t^T \int_s^T E_t[ D_s \sigma_u] \d u \d s.
\]
\end{proof}
Again, notice that Proposition \ref{prop: cpp} proves 
 Equations \eqref{eq: main theorem - fractional case} and \eqref{eq: main theorem - rough case} in Theorem \ref{th: main theorem} for the compound Poisson case.

\subsection{Extension to more general Lévy processes}
Let $S$ follow the dynamics given by Equation \eqref{model of study} and let $S^{\eps}$ be the approximator of $S$ introduced in Section \ref{sec: Short-time behavior of the ATM Implied Volatility level}. Then, the dynamics of $S^{\eps}$ follow Equation $\eqref{model of study}$ as well with Lévy process $L^{\eps}$. Recall that $L^{\eps}$ is a compound Poisson process  with drift and, moreover,  $S_t^{\eps} \to S_t$ in $L^2(\Omega)$ and a.s. uniformly in $t$ in compact sets. We denote by $I_t(k^*_t)$ the ATM implied volatility for $S$ and $I^{\eps}_t(k^{*,\eps}_t)$ the ATM implied volatility for $S^{\eps}$. Notice that Proposition \ref{prop: cpp} applies for $S^{\eps}$ and, in particular, we know that 
\[
\lim_{T \to t} \partial_k I^{\eps}_t(k^{*,\eps}_t) = \frac{c_1^{\eps}}{\sigma_t} + \lim_{T \to t} \frac{\rho}{\sigma_t (T-t)^2} \int_t^T \int_s^T E_t[D_s\sigma_u] \d u \d s
\]
if $H \geq 1/2$ and
\[
\lim_{T \to t} (T-t)^{1/2-H} \partial_k I^{\eps}_t(k^{*,\eps}_t) = \lim_{T \to t} \frac{\rho}{\sigma_t (T-t)^{3/2-H}} \int_t^T \int_s^T E_t[D_s\sigma_u] \d u \d s
\]
if $H < 1/2$, where $c_1^{\eps} = \int_{|y| > \eps} y \nu(\d y)$. This section then consists in proving a lemma that allow us to deduce a similar result for the infinite variation case. We have to show that, concerning the convergence of the implied volatilities, we can switch limits. The first step in order to be able to justify the exchange of limits is that one of the limits converges. The following result ensures us this property.

\begin{lema} \label{approximate skew}
    Let $I_t^{\eps}(k^{*,\eps}_t)$ be the ATM implied volatility associated to the model given by $S^{\eps}$ and $I_t(k^*_t)$ the ATM implied volatility associated to the model given by $S$. Then,
    \[
    \lim_{\eps \to 0} \partial_k I_t^{\eps}(k^{*,\eps}_t) = \partial_k I_t(k_t^*)
    \]
    uniformly on $T-t$.
\end{lema}
\begin{proof}
   We will do a similar strategy as with Lemma \ref{lemma: uniform convergence uncorrelated}. First, recall that
   \[
   I_t^{\eps}(k^{*,\eps}_t) = \Bac^{-1}(T,t,S_t^{\eps}, S_t^{\eps}, V_t^{\eps}), \quad V_t^{\eps} = E_t[(S_T^{\eps}- S_t^{\eps})_+]
   \]
   and
   \[
   I_t(k_t^*) = \Bac^{-1}(T,t,S_t,S_t, V_t), \quad V_t = E_t[(S_T - S_t)_+].
   \]
   Notice also that
   \[
   I_t^{\eps}(k_t^{*, \eps}) =   I_t^{\eps}(k_t^{*, \eps}) - \Bac^{-1}(T,t,S_t^{\eps},S_t, V_t^{\eps}) + \Bac^{-1}(T,t,S_t^{\eps},S_t, V_t^{\eps}),
   \]
   so
   \[
   \partial_k  I_t^{\eps}(k_t^{*, \eps}) = \partial_k I_t^{\eps}(k_t^{*, \eps}) - \partial_k\Bac^{-1}(T,t,S_t^{\eps},S_t, V_t^{\eps}) + \partial_k \Bac^{-1}(T,t,S_t^{\eps},S_t, V_t^{\eps}).
   \]
    Let's first analyze the difference
    \[
    \partial_k I_t^{\eps}(k_t^{*, \eps}) - \partial_k\Bac^{-1}(T,t,S_t^{\eps},S_t, V_t^{\eps}).
    \]
    Using the mean value theorem, we see that there exists a middle point $\xi$ between $S_t^{\eps}$ and $S_t$ (and due to the construction of $S_t^{\eps}$ we know that $S_t \leq \xi \leq S_t^{\eps}$) such that
    \[
     \partial_k I_t^{\eps}(k_t^{*, \eps}) - \partial_k\Bac^{-1}(T,t,S_t^{\eps},S_t, V_t^{\eps}) = \partial^2_{kk}\Bac^{-1}(T,t,S_t^{\eps},\xi, V_t^{\eps})(S_t^{\eps}-S_t).
    \]
    As we have seen before,
    \[
    \partial_k \Bac^{-1}(T,t,S_t^{\eps}, \xi, V_t^{\eps}) = \frac{-1}{\Phi(d(\xi, \Bac^{-1}(T,t,S_t^{\eps},\xi, V_t^{\eps}))}
    \]
    so
    \begin{align*}
        &\partial_{kk}^2 \Bac^{-1}(T,t,S_t^{\eps}, \xi, V_t^{\eps}) \\
        = & \frac{1}{\Phi(d(\xi,\Bac^{-1}(T,t,S^{\eps}_t,\xi,V_t^{\eps})))^2}\Phi'(d(\xi, \Bac^{-1}(T,t,S_t^{\eps}, \xi, V_t^{\eps}))\frac{\d}{\d k}d(k,\Bac^{-1}(T,t,S_t^{\eps},\xi, V_t^{\eps}))
    \end{align*}
    where $\frac{\d}{\d k}$ denotes the total derivative with respect to $k$. Notice that
    \[
    \frac{\d}{\d k}d(k,\sigma(k)) = \partial_k d(k, \sigma(k)) + \partial_{\sigma}d(k,\sigma(k)) \partial_k \sigma(k).
    \]
    Moreover,
    \[
    \partial_kd(k,\sigma(k)) = \frac{-1}{\sigma(k) \sqrt{T-t}}, \quad \partial_{\sigma} d(k,\sigma(k)) = \frac{k-x}{\sigma(k)^2\sqrt{T-t}}
    \]
    so
    \[
    \frac{\d}{\d k} d(k,\sigma(k)) = \frac{1}{\sqrt{T-t}}\left( \frac{-1}{\sigma(k)} + \frac{(k-x)\partial_k \sigma(k)}{\sigma(k)^2}\right).
    \]
   Therefore, applying the previous computations to Bachelier context we have that
    \begin{align*}
        &\frac{\d}{\d k} d(k,\Bac^{-1}(T,t, S_t^{\eps}, \xi, V_t^{\eps}))\\
        = &\frac{1}{\Bac^{-1}(T,t,S_t^{\eps}, \xi, V_t^{\eps})\sqrt{T-t}}\left(-1 + \frac{\xi - S_t^{\eps}}{(\Bac^{-1}(T,t,S_t^{\eps},\xi, V_t^{\eps})) \Phi(d(k,\Bac^{-1}(T,t,S_t^{\eps},\xi, V_t^{\eps})))}\right).
    \end{align*}
    First, using a slight modification of Proposition \ref{prop: level arbitrary Lévy} for our context we deduce that $\Bac^{-1}(T,t,S_t^{\eps}, \xi, V_t^{\eps}) \to I_t(k_t^*)$ as $\eps \to 0$ uniformly on $T-t$. Moreover, by the choice of our approximators,
    \[
    |\xi- S_t^{\eps}| \leq |S_t - S_t^{\eps}| \leq 2\eps.
    \]
    Finally, using that
    \[
    1 \leq \frac{1}{\Phi(d(\xi,\Bac^{-1}(T,t,S_t^{\eps}, \xi, V_t^{\eps})))} \leq 2
    \]
    we have that
    \[
    \lim_{\eps \to 0}\left|\frac{\d}{\d k} d(k,\Bac^{-1}(T,t, S_t^{\eps}, \xi, V_t^{\eps})) + \frac{1}{I_t(k_t^*))\sqrt{T-t}}\right| = 0
    \]
    uniformly on $T-t$. We can also interpret this limit as
    \[
    \lim_{\eps \to 0} \sqrt{T-t} \frac{\d}{\d k} d(\xi,\Bac^{-1}(T,t,S_t^{\eps}, \xi, V_t^{\eps})) = \frac{-1}{I_t(k_t^*)} 
    \]
    uniformly on $T-t$. For the remaining terms appearing in $\partial_{kk}^2\Bac^{-1}(T,t,S_t^{\eps}, \xi, V_t^{\eps})$, notice that
    \[
    1 \leq \frac{1}{\Phi(d(\xi, \Bac^{-1}(T,t,S_t^{\eps}, \xi, V_t^{\eps})))^2} \leq 4
    \]
    and $\Phi'(d(\xi, \Bac^{-1}(T,t,S_t^{\eps}, \xi, V_t^{\eps})))$ is uniformly bounded. In order to conclude this first step of the approximation argument, we derive the following estimate
    \begin{align*}
         &|\partial^2_{k,k}\Bac^{-1}(T,t,S_t^{\eps},\xi, V_t^{\eps})|\cdot |S_t^{\eps}-S_t| \\
          \leq &C\left(\sqrt{T-t} \left|\frac{\d}{\d k} d(\xi, \Bac^{-1}(T,t,S_t^{\eps}, \xi, V_t^{\eps}))\right| \right)|S_t - S_t^{\eps}|
    \end{align*}
    which tends to zero uniformly on $T-t$ as $\eps \to 0$. Up to this point, we have that
    \[
    \lim_{\eps \to 0} \left[ \partial_k I_t^{\eps}(k_t^{*, \eps}) - \partial_k I_t(k_t^*) \right]= \lim_{\eps \to 0} \left[ \partial_k \Bac^{-1}(T,t,S_t^{\eps}, S_t, V_t^{\eps}) - \partial_k I_t(k_t^*)\right].
    \]
    Now,
    \begin{align*}
       &\partial_k \Bac^{-1}(T,t,S_t^{\eps}, S_t, V_t^{\eps}) - \partial_k I_t(k_t^*) \\
       = &  \partial_k \Bac^{-1}(T,t,S_t^{\eps}, S_t, V_t^{\eps}) - \partial_k \Bac^{-1}(T,t,S_t, S_t, V_t^{\eps}) + \partial_k \Bac^{-1}(T,t,S_t, S,t, V_t^{\eps}) -     \partial_k I_t(k_t^*) \\
       = & A_1 + A_2.
    \end{align*}
    Regarding $A_1$, we resort to the mean value theorem to say that there exists $\xi$ between $S_t$ and $S_t^{\eps}$ such that
    \[
    |\partial_k \Bac^{-1}(T,t,S_t^{\eps}, S_t, V_t^{\eps}) - \partial_k \Bac^{-1}(T,t,S_t, S_t, V_t^{\eps})| = |\partial^2_{xk}\Bac^{-1}(T,t,\xi,S_t, V_t^{\eps})| |S_t -S_t^{\eps}|.
    \]
    Using the same arguments as with $|\partial^2_{kk}\Bac^{-1}(T,t,S_t^{\eps}, \xi, V_t^{\eps})|$, we can see that $A_1 \to 0$ as $\eps \to 0$ uniformly on $T-t$. For $A_2$, we use the mean value theorem again to say that there exists a middle point $\xi$ between $V_t$ and $V_t^{\eps}$ such that
    \[
   | \partial_k \Bac^{-1}(T,t,S_t, S,t, V_t^{\eps}) -     \partial_k I_t(k_t^*) |= |\partial^2_{\sigma k}\Bac^{-1}(T,t,S_t,S_t, \xi)| |V_t - V_t^{\eps}|.
    \]
     Now, from the fact that 
     \[
     \partial^2_{\sigma k}\Bac^{-1}(T,t,x,x,\sigma) = 0
     \]
     we see that $A_2$ vanishes. This proves then that
     \[
     \lim_{\eps \to 0} \partial_k I^{\eps}_t(k_t^{*.\eps}) = I_t(k_t^*) 
     \]
     uniformly on $T-t$.
    \end{proof}
    This lemma allows us to switch limits in the case where $\lim_{T \to t} I_t^{\eps}(k_t^{*,\eps})$ is well defined for almost every $\eps$. However, we have proved that when $H < 1/2$ the skew may exhibit a blowup, and has to be compensated with $(T-t)^{1/2-H}$ in order to obtain a finite limit when $T \to t$. Nevertheless, as a consequence of the previous lemma we can deduce the following corollary which deals with the $H < 1/2$ case.
    \begin{corol} \label{corol H < 1/2}
        Let $S$ follow the model \eqref{model of study} with $\sigma$ satisfying Hypotheses \ref{hyp 1}, \ref{hyp 2} and \ref{hyp 3} with $H \in (0,1/2)$. Then,
        \[
        \lim_{\eps \to 0} (T-t)^{1/2-H} \partial_k I_t^{\eps}(k_t^{*,\eps}) = (T-t)^{1/2 - H} \partial_k I_t(k_t^*)
        \]
        uniformly on $T-t$.
    \end{corol}
    Now we have all the necessary ingredients in order to prove a result concerning the short-time behavior of the skew when $L$ belongs to a wide class of pure-jump Lévy processes.
\begin{prop} \label{prop: inf act inf var}
      Let $S$ be a stock price with dynamics following model \eqref{model of study} with $L$ a pure-jump Lévy martingale satisfying that there exists a constant $c_1 \in \R$ such that $\lim_{\eps \to 0} c_1^{\eps} = c_1$. Then,
    \begin{itemize}
        \item If $H \geq 1/2$, then
        \[
        \lim_{T \to t} \partial_kI_t(k^*_t) = \frac{c_1}{\sigma_t} + \lim_{T \to t} \frac{\rho}{\sigma_t (T-t)^2} \int_t^T \int_s^T E_t[D_s\sigma_u] \d u \d s.
        \]
        \item If $H < 1/2$, then
        \[
        \lim_{T \to t} (T-t)^{1/2-H}\partial_k I_t( k^*_t) = \lim_{T \to t} \frac{\rho}{\sigma_t (T-t)^{3/2+H}}  \int_t^T \int_s^T E_t[ D_s \sigma_u] \d u \d s.
        \]
    \end{itemize}
\end{prop}
\begin{proof}
    We first deal with $H > 1/2$. Since $L_t^{\eps}$ is a compound Poisson process and $\lim_{\eps \to 0} c_1^{\eps} = c_1 \in \R$ we know that $\partial_kI_t^{\eps}(k_t^{*,\eps})$  has a finite limit as $\eps \to 0$ uniformly on $T-t$. Moreover, when $H \geq 1/2$ we have that the limit as $T$ tends to $t$ of $\partial_k I_t^{\eps}(k_t^{*,\eps})$ exists for almost every $\eps$. Hence, when $H \geq 1/2$ we can directly apply the Moore-Osgood theorem in order to justify that
    \[
    \lim_{T \to t} \lim_{\eps \to 0}\partial_k I_t^{\eps}(k_t^{*,\eps}) = \lim_{\eps \to 0} \lim_{T \to t} \partial_k I_t^{\eps}(k_t^{*,\eps}).
    \]
    Now, combining this result with Proposition \ref{prop: cpp} we have that
    \begin{align*}
        \lim_{T \to t} I_t(k_t^*) = &\lim_{T \to t } \lim_{\eps \to 0} \partial_k I_t^{\eps}(k_t^{*,\eps}) \\
        = & \lim_{T \to t} \lim_{\eps \to 0} \left( \frac{c_1^{\eps}}{\sigma_t} + \frac{\rho}{\sigma_t (T-t)^2}\int_t^T \int_s^T E_t(D_s\sigma_u) \d u \d s\right) \\
        = & \lim_{\eps \to 0} \frac{c_1^{\eps}}{\sigma_t} + \lim_{T \to t} \frac{\rho}{\sigma_t(T-t)^2}\int_t^T \int_s^T E_t(D_s\sigma_u) \d u \d s \\
        = & \frac{c_1}{\sigma_t} + \lim_{T \to t} \frac{\rho}{\sigma_t(T-t)^2}\int_t^T \int_s^T E_t(D_s\sigma_u) \d u \d s.
    \end{align*}
    The case $H < 1/2$ follows the same lines as the case $H \geq 1/2$ but using Corollary \ref{corol H < 1/2} instead of Lemma \ref{approximate skew}.
\end{proof}
The proof of the main result of this paper, Theorem \ref{th: main theorem}, can now be done in one line.
\begin{proof}[Proof of Theorem \ref{th: main theorem}]
    On the one hand, Equation \eqref{eq: main theorem - level} is proved in proposition \ref{prop: level arbitrary Lévy}. On the other hand, Equations \eqref{eq: main theorem - fractional case} and \eqref{eq: main theorem - rough case} are a consequence of Proposition \ref{prop: inf act inf var}.
\end{proof}

\section{Examples and numerical analysis} \label{section: examples and numerical analysis}
The results from the previous sections have provided us with general formulas that we will develop in different particular scenarios. Moreover, we will perform computational experiments to test the veracity of these formulas for particular examples of stochastic volatilities and Lévy processes. Recall that the stock price $S$ is assumed to follow the equation
\[
S_t = S_0 + \int_0^t \sigma_r (\rho \d W_r + \sqrt{1-\rho^2} \d B_r) + L_t,
\]
where $L$ is a pure jump Lévy process with drift so that $L$ is a martingale and $\sigma$ satisfies Hypotheses \ref{hyp 1}, \ref{hyp 2} and \ref{hyp 3}. 

We will discuss 3 examples in this section.

\begin{enumerate}
    \item The first one is the generalized Bachelier-Bates model, that works under the assumption that $\sigma_t$ is a stochastic process and $L$ is a compound Poisson process. We will be looking at $\sigma$ following the Fractional Bergomi model, that is the natural generalization of the rough Bergomi model (see \cite{bayer2016pricing}) allowing $H \in (0,1)$ and not only $H < 1/2$. We will work with different Hurst indices $H$ in order to see reflected in the numerics all the factors involved in Theorem \ref{th: main theorem}.
    \item As a second example, we will let $\sigma$ follow again the Fractional Bergomi model and we will let $L$ be a CGMY Lévy process in order to allow the trajectories of $L$ to have infinite variation. In this example we will test some of the experiments we have performed to the Bachelier-Bates model for this scenario. The objective of this second example will be to detail the formulas in Theorem \ref{th: main theorem} for the case of a family of Lévy process with infinite activity and infinite variation paths. Moreover, we will test the veracity of the formulas with numerical experiments.
    \item  Finally, we will take a look at the Normal Inverse Gaussian process. This process does not satisfy the hypotheses of Theorem \ref{th: main theorem} regarding the skew because $\lim_{\eps \to 0} c_1^{\eps} = \infty$ (whenever the parameters are not chosen so that the law of the jumps of $L$ is symmetric). In such case, we will see that the failure in the applicability of Theorem \ref{th: main theorem} comes from the non-differentiability of the implied volatility surface in the ATM scenario.
\end{enumerate}  For the first two examples, the ATM-IV level will be computed as follows: since $k_t^* = S_t$ then
\[
E_t\left( (S_T - S_t)_+ \right) = \Bac(T,t,S_t, k_t^*, I_t(k_t^*)) = \frac{I_t(k_t^*)}{\sqrt{2\pi}} \sqrt{T-t}
\]
and therefore
\[
I_t(k_t^*) = \frac{E_t\left( (S_T - S_t)_+ \right) \sqrt{2\pi}}{\sqrt{T-t}}.
\]
In order to simulate the ATM-IV Skew, we will use the formula
\[
\partial_k I_t(k_t^*) = \frac{1/2 - E_t\left( \1_{\{S_T \geq S_t\}}\right)}{\sqrt{\frac{T-t}{2\pi}}}.
\]

 In order to perform the simulations for the last example, we won't rely on the ATM-IV level formula because we will explore what happens when we perturb the strike around $k_t^*$. In that case, we will use the algorithm in \cite{jackel2017implied} in order to simulate options OTM and ITM The code that we have utilized for this section can be found in the GitHub repository \url{https://github.com/oscarbures01/Bachelier-JD-SV-IV}.

\subsection{The generalized Bachelier-Bates model}
The generalized Bachelier-Bates model assumes that the stock price process $S$ follows the model \eqref{model of study} with $L$ a compound Poisson process with drift. For this case, we can write $L$ as the sum of a pure-jump process and a drift term. Hence, $S$ can be written as
\begin{equation} \label{model - poisson case}
S_t = S_0 - c_1 t + \int_0^t \sigma_r (\rho \d W_r + \sqrt{1-\rho^2} \d B_r) + Z_t,
\end{equation}
where $Z$ is a compound Poisson process with Lévy measure $\nu$ and $c_1 = \int_{\R} y \nu(\d y)$. This model is called the generalized Bachelier-Bates model because the dynamics are assumed to drive the price, not the log-price and we let $\sigma$ be a general stochastic process unlike the classical Bates model, where the volatility is assumed to follow the Heston model. Usually, the law of the jumps of $L$ is assumed to be a Gaussian distribution or a double exponential distribution (see \cite{tankov2003financial}). We will work with both distributions in order to see that the effect of the jumps in the level and the skew of the ATM implied volatility when the time to maturity is small does not depend on the law of the jumps as long as the parameters of the law are adjusted so that $c_1$ remains the same.
\subsubsection{Fractional Bergomi volatility with Gaussian/Laplace jumps}
We say that $S$ follows the stochastic volatility Fractional Bergomi dynamics if the volatility process $\sigma$ satisfies
\begin{equation} \label{fractional bergomi}
\sigma_t^2 = \sigma_0^2 \exp \left( \alpha W^H_t - \frac{1}{2} \alpha^2 t^{2H} \right)
\end{equation}
where $W^H$ is a fractional Brownian motion of Hurst index $H \in (0,1)$. For the particular case $H < 1/2$, the Model \eqref{fractional bergomi} coincides with the rough Bergomi model. Notice that, if we write $W^H_t = \int_0^t K_H(t,s) \d W_s$ for the kernel described in \cite{nualart2006malliavin} then
\begin{equation} \label{rough derivative}
D_s\sigma_u = \frac{1}{2}\sigma_u K_H(u,s)\1_{[0,u]}(s).
\end{equation}

\subsubsection{The case $H \geq 1/2$}

Regarding the level, we know due to Theorem \ref{th: main theorem} that
\begin{equation} \label{ATM-IV level theoretical}
\lim_{T \to 0} I_0(k^*) = \sigma_0.
\end{equation}
Moreover, regarding the skew, we can show as in \cite{alos2023implied} that when $H \geq 1/2$, then
\[
\lim_{T \to 0} \frac{\rho}{\sigma_0 T^2 }\int_0^T \int_s^T \E[D_s \sigma_u] \d u \d s = \begin{cases}
    0 & H > 1/2 \\
    \frac{\rho \alpha}{4} & H = 1/2.
\end{cases}
\]
In virtue of Theorem \ref{th: main theorem} we have
\begin{equation} \label{ATM-IV skew theoretical}
\lim_{T \to 0} \partial_k I_0(k^*) = \begin{cases}
    \frac{c_1}{\sigma_0} & H > 1/2 \\
    \frac{c_1}{\sigma_0} + \frac{\rho \alpha}{4} & H =1/2
\end{cases}
\end{equation}
where $c_1$ depends on the law of the jumps. For instance, if we assume $Z_t = \sum_{i = 1}^{N_t} X_i$ where $N$ is a Poisson process of intensity $\lambda$ and $X_i$ are i.i.d. random variables with $\delta = \E[X_i]$. Then, $c_1 = \lambda \delta$.
  
\begin{ex} \label{ex: Bates H > 1/2}
We will numerically show that Equations \eqref{ATM-IV level theoretical} and \eqref{ATM-IV skew theoretical} hold when we fix the law of the jumps (and therefore $\lambda \delta$ is fixed) and we let $\sigma_0$ vary. The parameters we have selected for the simulations are 
\[
T = 10^{-5}, \quad S_0 = 10, \quad \alpha = 0.5, \quad \rho = -0.3, \quad H = 0.7.
\]
In order to simulate the effect of the jumps, we assume that the Lévy process $L$ is a compound Poisson process with intensity $\lambda = 5$ and law of the jumps $N(0.01, 0.2)$, so that $c_1 = \lambda \delta = 5 \times 0.001 = 0.05$.

To check the behavior of the ATM-IV level and skew as we change $\sigma_0$, we will perform simulations for $\sigma_0 \in \{0.1, \dots, 1.4\}$. The simulations are done with a  Monte Carlo method with $2$ million paths simulated with antithetic variables. 
\begin{figure}[H]
        \centering
        \includegraphics[width = 13cm, height = 6cm]{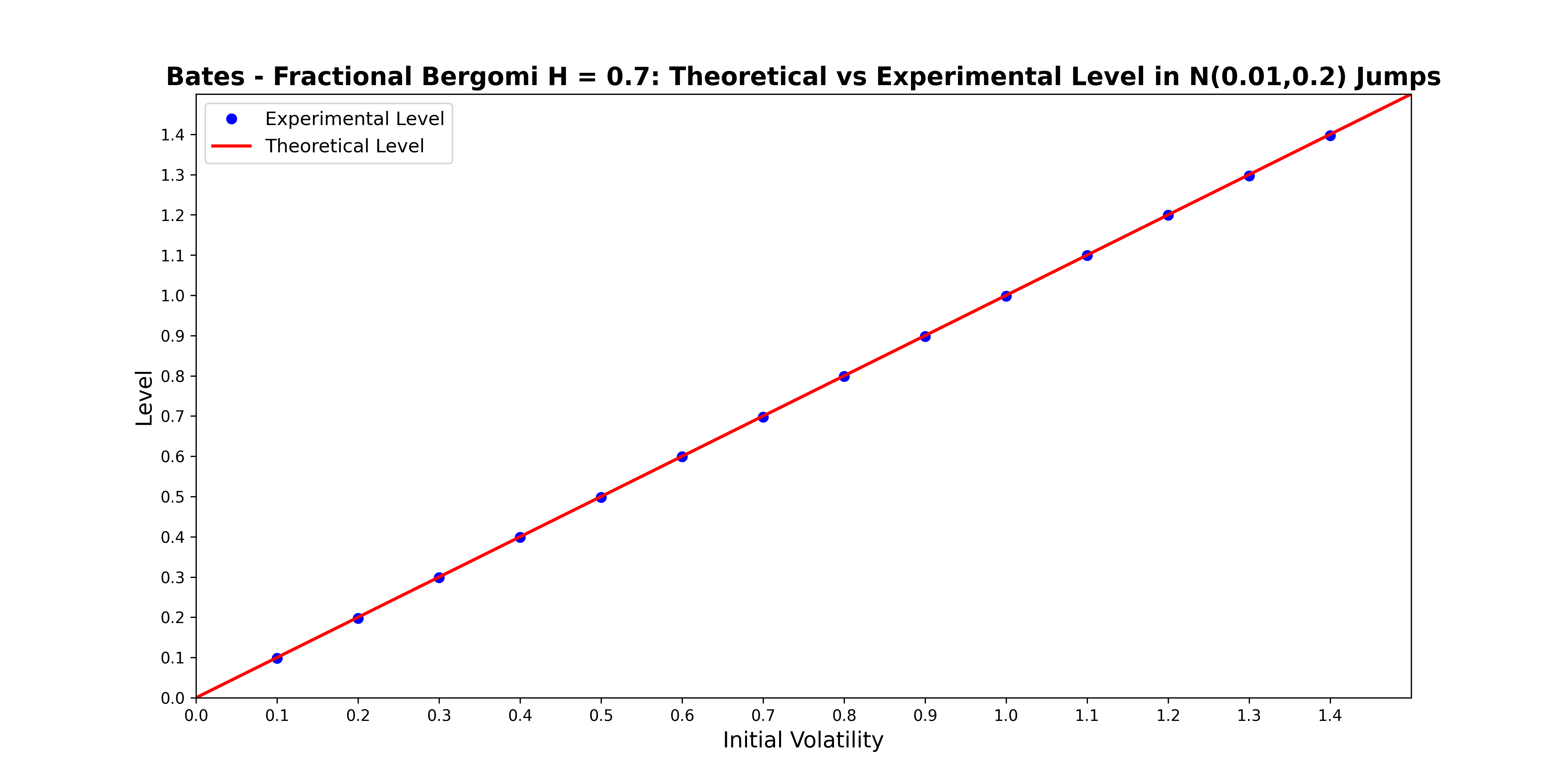}
        \caption{In red we can see the plot of the identity map $(\sigma_0,\sigma_0)$. In blue we can see how the simulated pairs $(\sigma_0, I_0(k^*))$ with $T = 10^{-5}$ for $\sigma_0 \in \{0.1, \dots, 1.4\}$ are overlapped with the theoretical results (plotted in red).}
        \label{fig: Bates level Gaussian jumps}
    \end{figure}
In order to simulate the ATM-IV skew, we choose the same set of parameters as with the ATM-IV level. Notice that Equation \eqref{eq: main theorem - fractional case} in Theorem \ref{th: main theorem} applied to $\sigma$ following the dynamics given by \eqref{fractional bergomi} with our set of parameters states that
\[
\lim_{T \to 0} \partial_k I_0(k^*) = \frac{c_1}{\sigma_0}.
\]
Hence, we want to check that the dependence of the ATM-IV skew with respect to $\sigma_0$ when $c_1$ is fixed is reciprocal to $\sigma_0$. To do so, we let
\begin{equation} \label{eq: Gaussian jumps}
Z_t = \sum_{i=1}^{N_t} X_i, \quad N_t \sim Poiss(5 t), \quad X_i \sim N(\delta, 0.2) \text{ with }\delta \in \{-0.01, 0\},
\end{equation}
so $c_1 = \lambda \delta \in \{-0.05, 0\}$. In order to stress that the effect the jumps have on the ATM-IV skew is encapsulated in $c_1$ and not on the specific law of the jumps, we now simulate the ATM-IV skew with jumps modeled as
\begin{equation} \label{eq. Laplace jumps}
Z_t = \sum_{i=1}^{N_t} X_i, \quad N_t \sim Poiss(5t), \quad X_i \sim L(0.01, 1).
\end{equation}
Recall that the Laplace or Double Exponential distribution $L(\delta, b)$ is a family of absolutely continuous laws with probability density function
\[
f(x; \delta, b) = \frac{1}{2b} \exp \left( - \frac{|x-\delta|}{b} \right).
\]
In this case, we have $\E[X_i] = 0.01$, so $c_1 = 0.05$. In Figure \ref{fig: Bates model different means} we can see the ATM-IV skew when the compound Poisson process is of the form \eqref{eq: Gaussian jumps} and \eqref{eq. Laplace jumps}.
\begin{figure}[H]
        \centering
        \includegraphics[width = 13cm, height = 6cm]{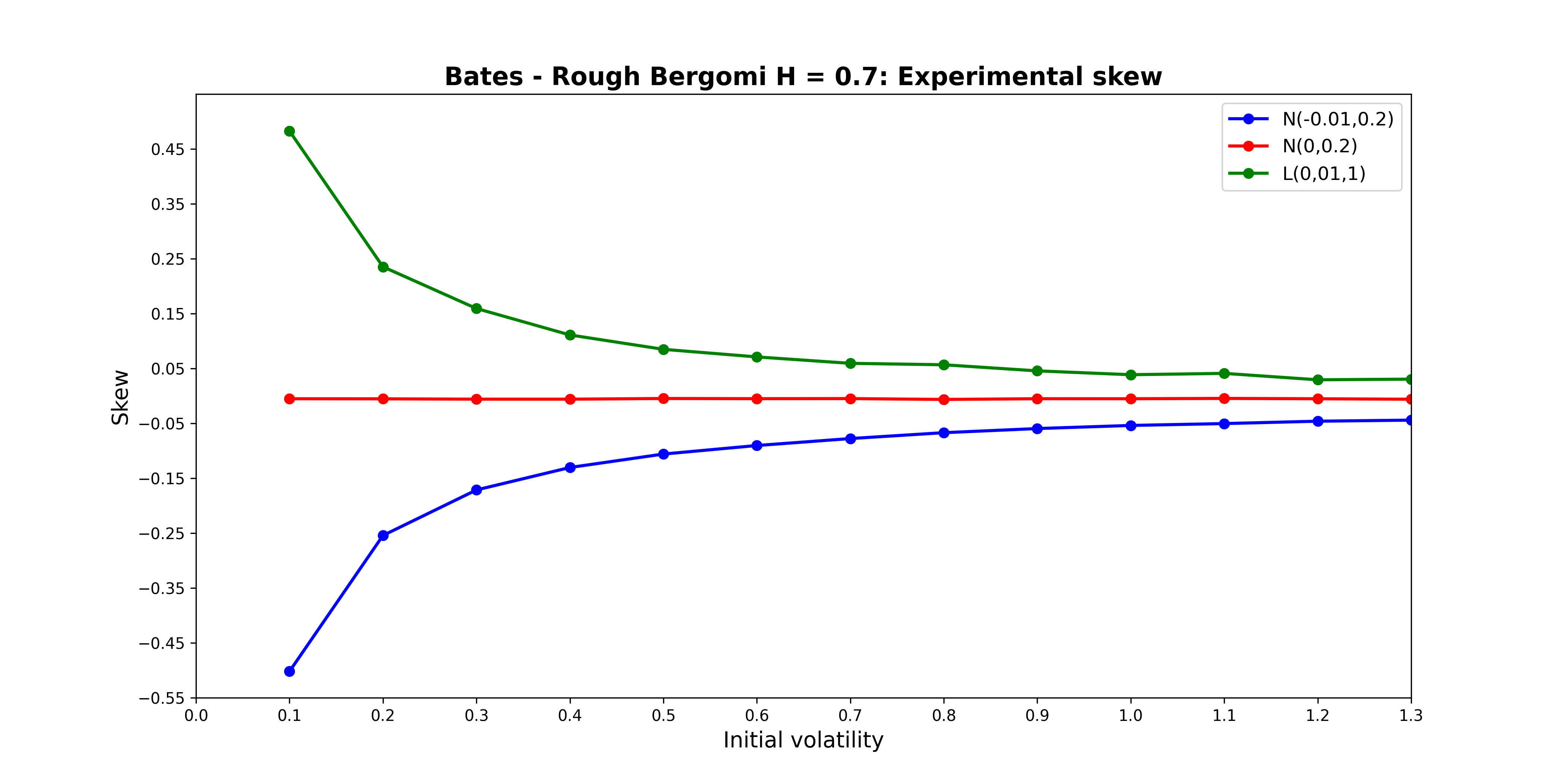}
        \caption{In blue, we can see the plot of the simulated map $\sigma_0 \mapsto \partial_kI_0(k^*)$ in the case where the jumps follow a $N(-0.01, 0.2)$ law. In red, we can see the plot of the simulated map $\sigma_0 \mapsto \partial_kI_0(k^*)$ in the case were the jumps follow a $N(0,0.2)$ law. In green  we can see the plot of the simulated map $\sigma_0 \mapsto \partial_kI_0(k^*)$ in the case were the jumps follow a $L(0.01,1)$ law.}
        \label{fig: Bates model different means}
    \end{figure}
    It can be observed in Figure \ref{fig: Bates model different means} that, apart from the reciprocal dependence on the initial volatility, the curve with Gaussian jumps of mean $-0.01$ is symmetric to the curve with Laplace jumps of mean $0.01$. This highlights that the behavior of the skew when we fix $\sigma_0$ depends on the mean of the jumps, rather than on their specific distribution.
    \end{ex}
\subsubsection{The case $H < 1/2$}
In the case $H < 1/2$, we can rely on the computations done in \cite{alos2023implied} to compute the formulas in Theorem \ref{th: main theorem}  for the generalized Bachelier-Bates model. If we let $H < 1/2$ and  $Z$ be a compound Poisson process with intensity $\lambda$ and mean of the jumps $\delta$. Then,
    \begin{equation} \label{ATM-IV level and skew theoretical}
    \lim_{T \to 0} I_0(k^*) = \sigma_0 \quad \text{ and } \quad \lim_{T \to 0} T^{1/2-H}\partial_k I_0(k^*) =  \frac{2\rho \alpha \sqrt{2H}}{(3+4H(2+H))}.
    \end{equation}
    Moreover,
    \[
     \lim_{T \to 0} \partial_k I_0(k^*) = \begin{cases}
        \infty & \rho > 0 \\
        -\infty & \rho < 0.
     \end{cases}
    \]

\begin{ex} In this example we will perform numerical simulations to show that the equalities presented in Equation \eqref{ATM-IV level and skew theoretical} hold. Regarding the ATM-IV level we select the parameters
\[
T = 0.001, \quad S_0 = 100, \quad \alpha = 0.5, \quad \rho = -0.3, \quad H = 0.4,
\]
and we assume that the Lévy process is a compound Poisson process with intensity $\lambda = 5$ and the law of the jumps is a $L(0.1,0.1)$ distribution, so $c_1 = \lambda \delta  =0.5$.

Regarding the behavior of the ATM-IV skew as a function of the time to maturity $T$, we choose the same parameters for $\sigma$ and we assume that the Lévy process is a compound Poisson process with intensity $\lambda = 5$ and jumps following a $N(\delta, 0.2)$ distribution with $\delta \in \{-0.1, 0, 0.1\}$. In Figure \ref{fig: level and skew bates laplace & gaussian} we plot the results of the simulations.
\begin{figure}[H]
    \centering
    \subfigure[]{\includegraphics[width = 7cm, height = 4cm]{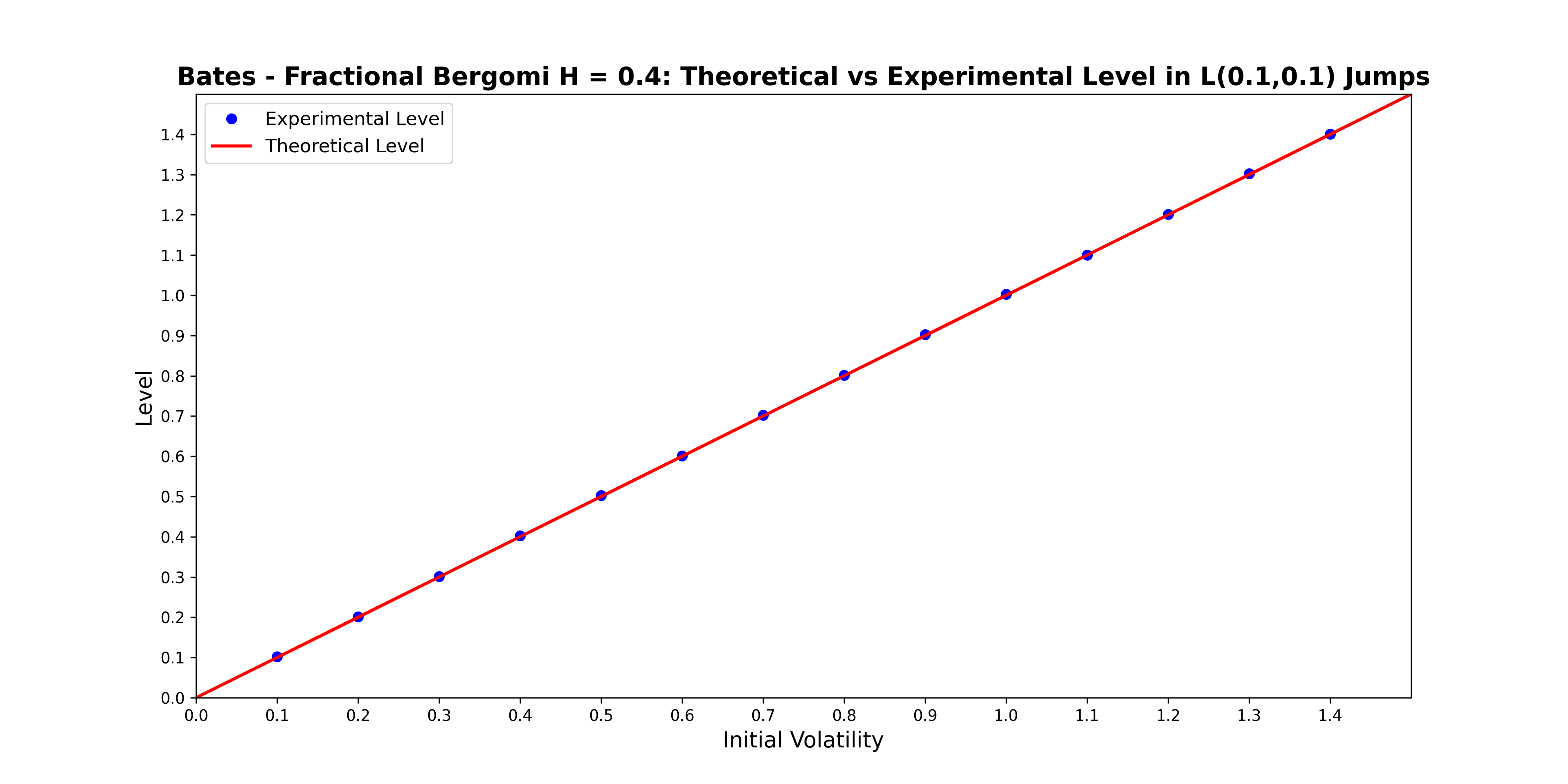}}
    \subfigure[]{\includegraphics[width = 7cm, height = 4cm]{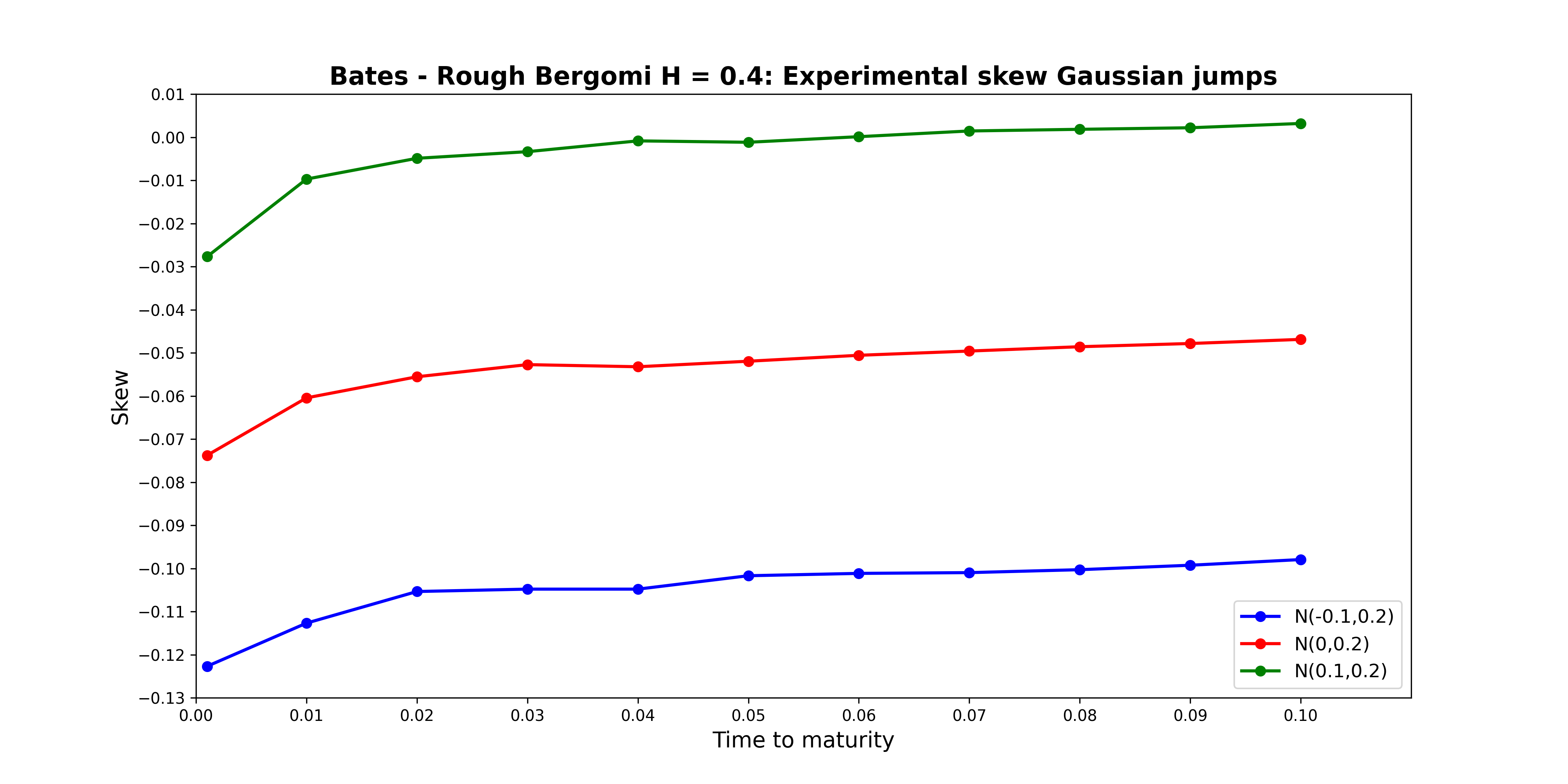}}
    \caption{(a) Theoretical vs experimental ATM implied volatility when the jumps follow a $L(0.1, 0.1)$ distribution. (b) Experimental skew when the jumps follow a $N(\delta, 0.2)$ distribution with $\delta \in \{-0.1, 0, 0.1\}$}
    \label{fig: level and skew bates laplace & gaussian}
\end{figure}
On the one hand, it is clear that the simulated ATM-IV level is perfectly aligned to the theoretical level. On the other hand, regarding the skew we can easily see that, as we change $\delta$ we obtain the same curve but shifted vertically. However, the fact that the chosen means $\delta$ are spread prevent us from appreciating any blow-up.

Let's take a closer look to the behavior when we change the mean of the jumps to make show in a more explicit way the blow-up effect. To do so, we fix $\sigma_0 = 0.3$ and we keep the same parameters for $\sigma$ and the same intensity of the compound Poisson process as in the ATM-IV level simulation. The law of the jumps is now assumed to be $N(\delta, 0.2)$ with $\delta \in \{-0.001, 0, 0.001\}$. In Figure \ref{fig: Bates model Gaussian jumps zoom} we can see more explicitly the tendency to $-\infty$ of the ATM-IV skew and we can also appreciate the fact that a modification on $c_1$ leads to a vertical translation of the graph of the ATM-IV skew as a function of $T$.
\begin{figure}[H]
        \centering
        \includegraphics[width = 14cm, height = 6.5cm]{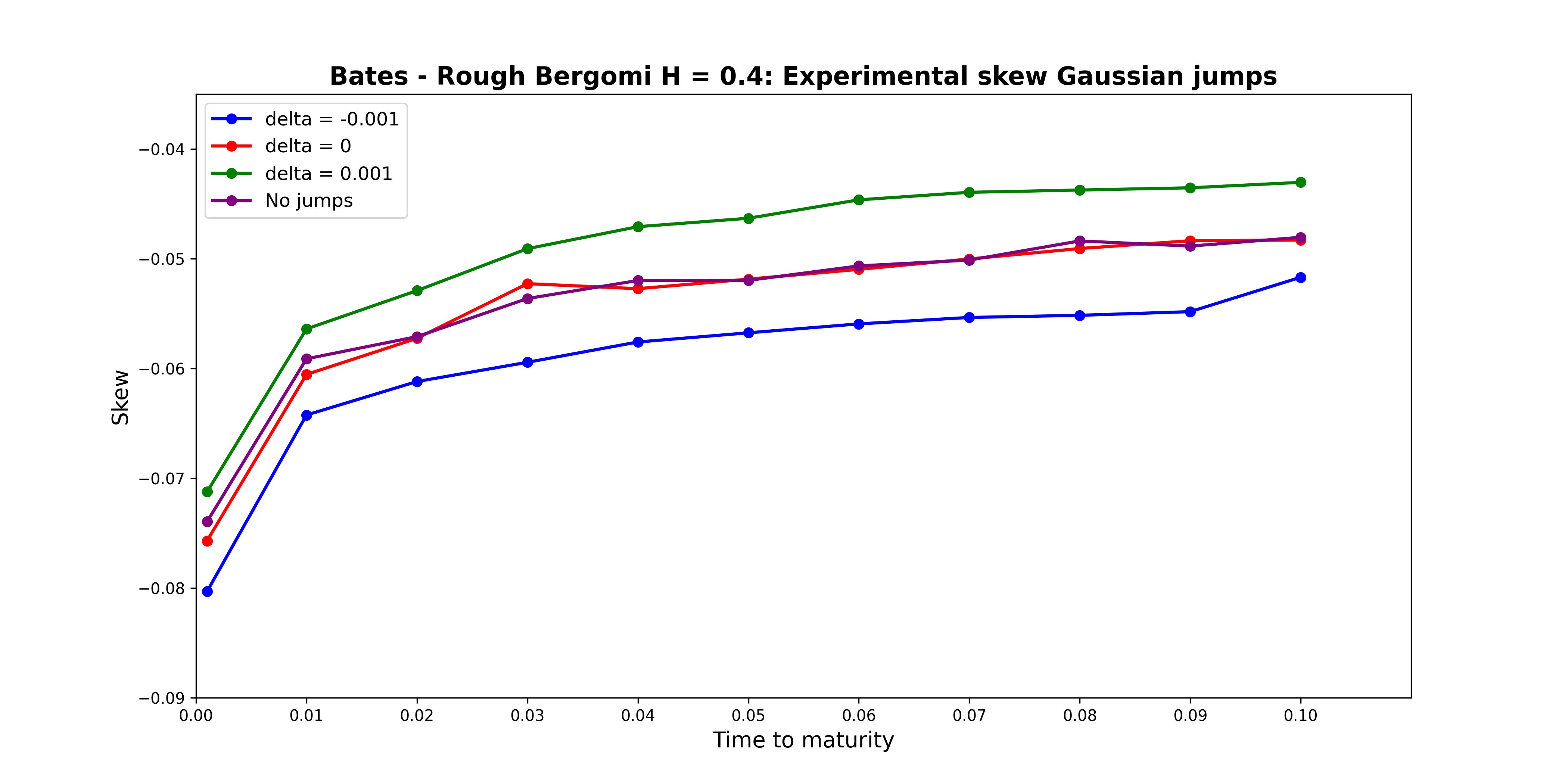}
        \caption{Plot of the maps $T \mapsto \partial_kI_0(k^*)$. In blue, we plot the map in the case where the jumps follow a $N(-0.001, 0.2)$ law ($c_1 = -0.005$). In green we plot the map in the case where the jumps follow a $N(0.001, 0.2)$ law ($c_1 = 0.005$). In red and purple we plot the map in the case where the jumps follow a $N(0,0.2)$ law and we have absence of jumps respectively ($c_1 = 0$).}
        \label{fig: Bates model Gaussian jumps zoom}
    \end{figure}
Notice that both the red and the purple curves in Figure \ref{fig: Bates model Gaussian jumps zoom} (that are the ones simulated with centered jumps and no jumps involved respectively) coincide with the curve obtained in \cite{alos2023implied}, where the authors study the asymptotic behavior of the ATM-IV level and skew for the stochastic volatility Bachelier model without jumps.

In order to produce all the images of this section, we have used a Monte Carlo method with 2 million simulated paths and antithetic variables.
\end{ex}
\subsection{Stochastic volatility with CGMY jumps}
For this example, we will assume that $S$ follows the model \eqref{model of study}, that is,
\[
S_t = S_0 + \int_0^t \sigma_r (\rho \d W_r + \sqrt{1-\rho^2} \d B_r) + L_t
\]
where $L$ is a pure jump Lévy martingale with Lévy measure
\[
\nu(\d y) = \begin{cases}
    C\frac{1}{|y|^{1+Y}}\exp\left( - G|y|\right)\d y & y < 0 \\
    C \frac{1}{y^{1+Y}}\exp(-M y) \d y & y > 0
\end{cases}
\]
where $C,G,M > 0$ and $Y \in (-\infty,2)$. This Lévy process is a particular case of a larger class of processes, named Tempered Stable Lévy processes (see \cite{ROSINSKI2007677} for a broader insight on the subject). Since for $Y \in (1,2)$ the trajectories of the process are of infinite variation, the application of Theorem \ref{th: main theorem} (and therefore, its usage to compare if the simulated results coincide with the theoretical ones) is subject to the existence of a constant $c_1$ such that $\lim_{\eps \to 0} c_1^{\eps} = c_1$. As a consequence, a preliminary step to the numerical analysis consists in discussing the existence of $c_1$ in the case $Y \in (1,2)$. Consider the Lévy measures $\nu_{\eps}(\d y)$ defined by
\[
\nu_{\eps}(\d y) = \begin{cases}
    C\frac{1}{|y|^{1+Y}}\exp\left( - G|y|\right) \d y & y < -\eps \\
    C \frac{1}{y^{1+Y}}\exp(-M y) \d y & y > \eps.
\end{cases}
\]
It is clear that $y \in L^1(\R \backslash [-\eps,\eps])$ so the constants $c_1^{\eps}$ are well defined. In order to compute $c_1^{\eps}$ notice that
\[
c_1^{\eps} = \int_{|y| > \eps} y \nu(\d y) = I_1(\eps) + I_2(\eps)
\]
where
\[
I_1(\eps) = \int_{\eps}^{\infty} \frac{1}{y^{Y}}\exp(-My) \d y
\]
and
\[
I_2(\eps) = -\int_{-\infty}^{-\eps} \frac{1}{(-y)^Y}\exp(Gy) \d y.
\]
For $I_1(\eps)$ we can perform the cange of variables $t = Mx$ in order to get
\[
I_1(\eps) = \int_{M\eps}^{\infty} \frac{M^{Y-1}}{t^Y}e^{-t} \d t = M^{Y-1} \Gamma(1-Y, M\eps),
\]
where 
\[
\Gamma(z,x) = \int_x^{\infty} t^{z-1}e^{-t} \d t
\]
denotes the upper incomplete Gamma function. Notice that 
\[
\lim_{\eps \to 0} M^{Y-1}\Gamma(1-Y, M\eps) = M^{Y-1}\Gamma(1-Y)
\]
as long as $Y \neq 1$ due to the monotone convergence theorem. If we study $I_2(\eps)$, we can repeat previous computations in order to get
\[
I_2(\eps) = - G^{Y-1}\Gamma(1-Y, G\eps).
\]
Hence, if $Y \neq 1$, then we have that
\[
\lim_{\eps \to 0} c_1^{\eps} =  \left(M^{Y-1} - G^{Y-1}\right)\Gamma(1-Y)=:c_1.
\]
\begin{remark}
    Notice that in the case $Y = 1$ we can still argue that $c_1$ exists in the case where $G = M$. Indeed, in this case we have $c_1^{\eps} = 0$ for all $\eps > 0$ and we can define $c_1 = 0$.
\end{remark}
We will study this two cases (symmetric $Y = 1$ and general $Y \in (1,2)$) and we will test them with fractional stochastic volatility models. Both cases are performed using a Monte Carlo method with $2$ million paths simulated with antithetic variables.

\subsubsection{The symmetric $Y = 1$ case}
Since we are forced to choose $G = M$, we will choose, for instance, $(C,G,M,Y) = (1,5,5,1)$ so that
\[
\nu(\d y) = \begin{cases}
    \frac{1}{|y|^2}\exp(-5|y|), & y < 0 \\
    \frac{1}{y^2} \exp( - 5y), & y > 0
\end{cases}
\]
and $c_1^{\eps} = c_1 = 0$. Since $c_1 = 0$, the objective of this section will be showing that not only Theorem \ref{th: main theorem} holds in the case where the Lévy process $L$ has infinite variation trajectories but also checking that the numerical results in this case match the numerical results when we don't consider jumps in the model. An application of Theorem \ref{th: main theorem} for our case allow us to derive the following conclusion.
\begin{ex}
    Let $S$ follow model \eqref{model of study} with $\sigma$ satisfying Equation \eqref{fractional bergomi}. Assume that $L$ is a CGMY process with $(C,G,M,Y) = (1,5,5,1)$. Then,
    \[
        \lim_{T \to 0}  I_0(k^*) = \sigma_0 \quad \text{and} \quad \lim_{T \to 0}T^{\max(1/2-H, 0)}\partial_kI_0(k^*) = \begin{cases}
    0 & H > 1/2 \\
    \frac{\rho \alpha}{4} & H = 1/2 \\
    \frac{2 \rho \alpha \sqrt{2H}}{3+4H(2+H)} & H < 1/2.
     \end{cases}
    \]  
\end{ex}
If we compare this formulas with the ones in \cite{alos2023implied} we observe that the contribution of the CGMY jumps in the symmetric case is null, since the formulas coincide with the case where there are no jumps. We will explore numerically the fact that having no jumps and having symmetric CGMY jumps leads to the same short-time ATM-IV behavior. Concerning the level, we simulate the process $S$ with $\sigma$ satisfying the Rough Bergomi dynamics \eqref{fractional bergomi} with parameters
\[
T = 0.001, \quad S_0 = 100, \quad \alpha = 0.5, \quad \rho = -0.3, \quad H = 0.4,
\]
and we simulate the level in the range $\sigma_0 \in \{0.1, \dots, 1.4\}$. For the CGMY jumps, we choose $(C,G,M,Y) = (1,5,5,1)$ as previously mentioned. In Figure \ref{fig: CGMY level symmetric} we can see how the experimental level is adjusted perfectly with the theoretical level.

\begin{figure}[H]
        \centering
        \includegraphics[width = 13cm, height = 6cm]{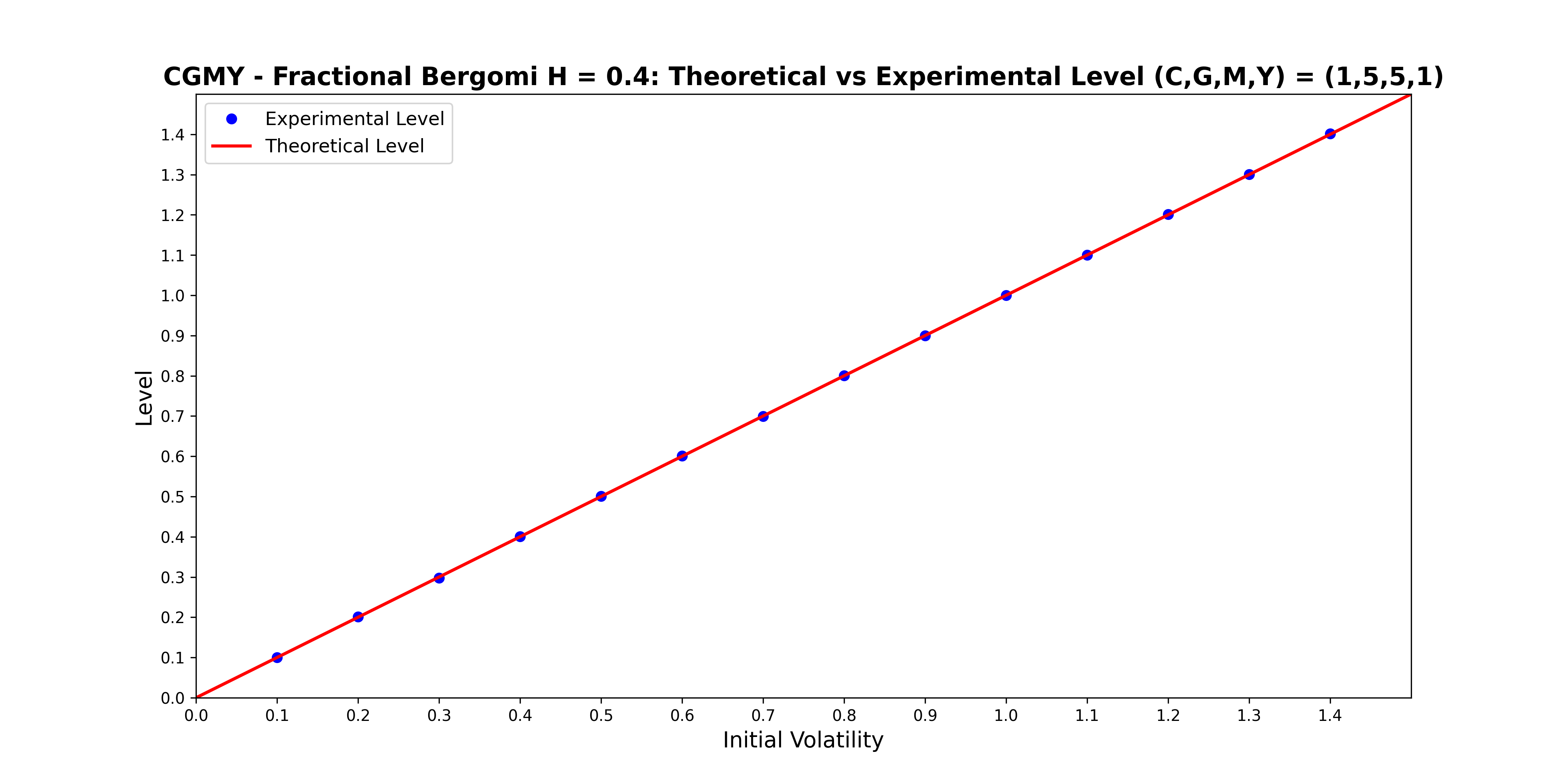}
        \caption{Overlapping of simulated pairs $(\sigma_0, I_0(k^*))$ (blue) with the theoretical curve of results $(\sigma_0,\sigma_0)$ (red).}
        \label{fig: CGMY level symmetric}
    \end{figure}

Regarding the skew, we consider the parameters 
\[
T = 0.001, \quad S_0 = 100, \quad \alpha = 0.5, \quad \rho = -0.3,
\]
and the parameters for the law of the jumps are $(C,G,M,Y) = (0.005, 5,5,1)$. We test Equations \eqref{eq: main theorem - fractional case} and \eqref{eq: main theorem - rough case} of Theorem \ref{th: main theorem} for $H \in \{0.4, 0.5, 0.7\}$ and $\sigma_0 \in \{0.1, \dots, 1.4\}$. We can see in Figures \ref{fig: CGMY skew symmetric 0.7}, \ref{fig: CGMY skew symmetric 0.5} and \ref{fig: CGMY skew symmetric 0.4} the successful results of the simulations of the ATM-IV skew functions of $\sigma_0$. 

\begin{figure}[H]
        \centering
        \includegraphics[width = 13cm, height = 6cm]{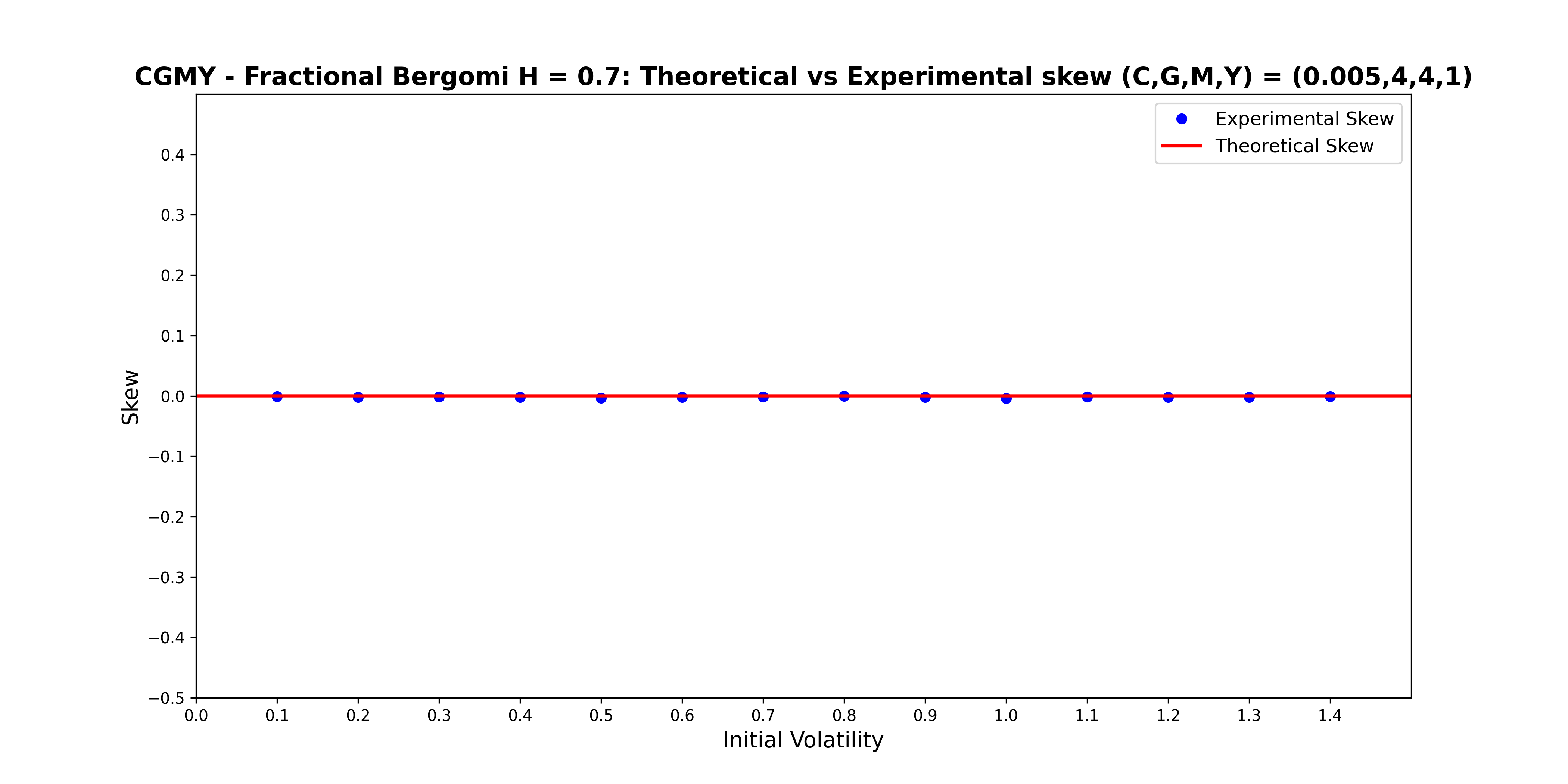}
        \caption{In blue we find the simulated values $(\sigma_0, \partial_k I_0(k^*))$ and we observe that they coincide with the theoretical curve $(\sigma_0, 0)$ plotted in red.}
        \label{fig: CGMY skew symmetric 0.7}
    \end{figure}

    \begin{figure}[H]
        \centering
        \includegraphics[width = 13cm, height = 6cm]{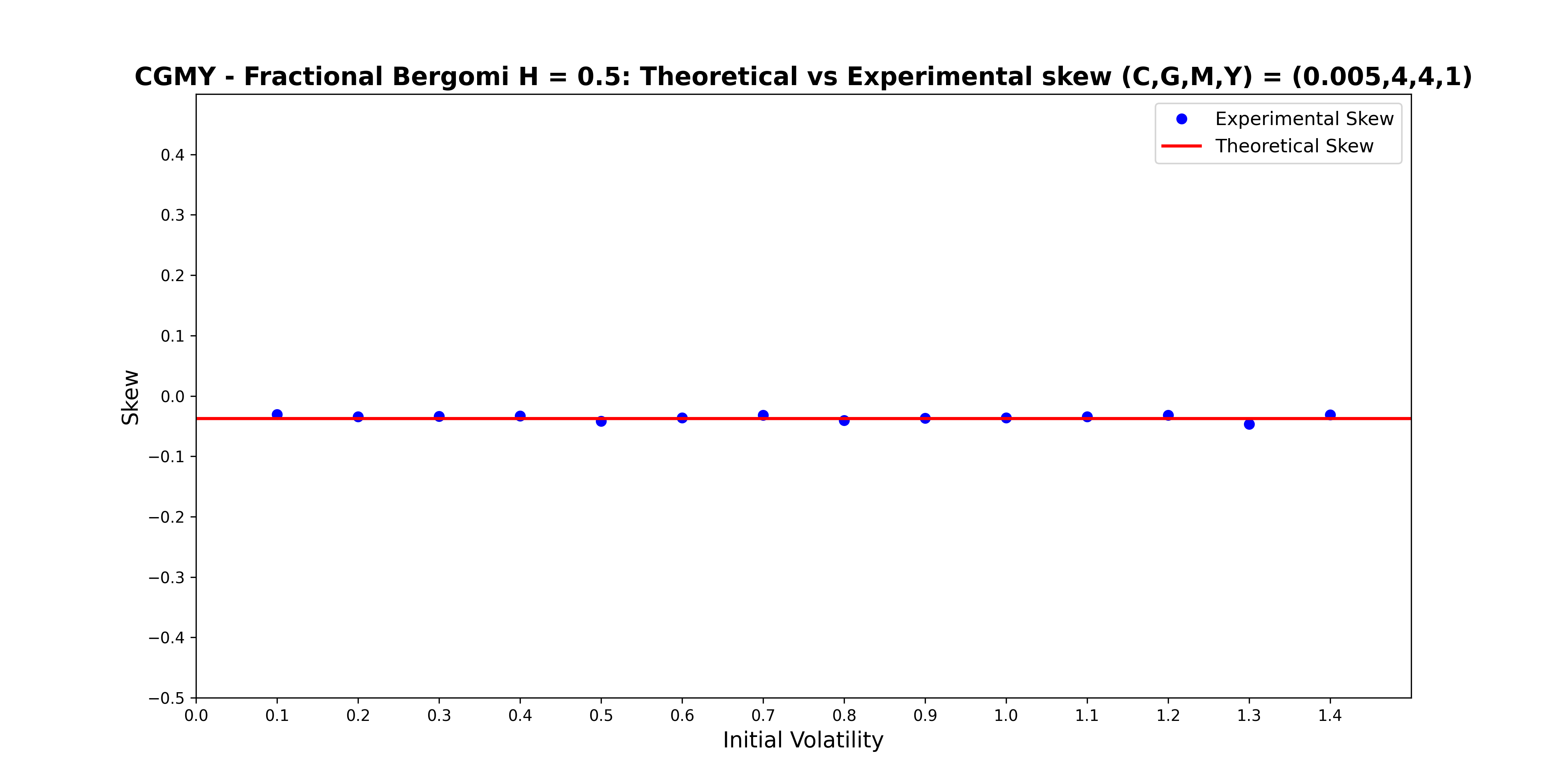}
        \caption{In blue we find the simulated values $(\sigma_0, \partial_k I_0(k^*))$ and we observe that they coincide with the theoretical curve $(\sigma_0, \frac{\rho \alpha}{4})$ plotted in red.}
        \label{fig: CGMY skew symmetric 0.5}
    \end{figure}

    \begin{figure}[H]
        \centering
        \includegraphics[width = 13cm, height = 6cm]{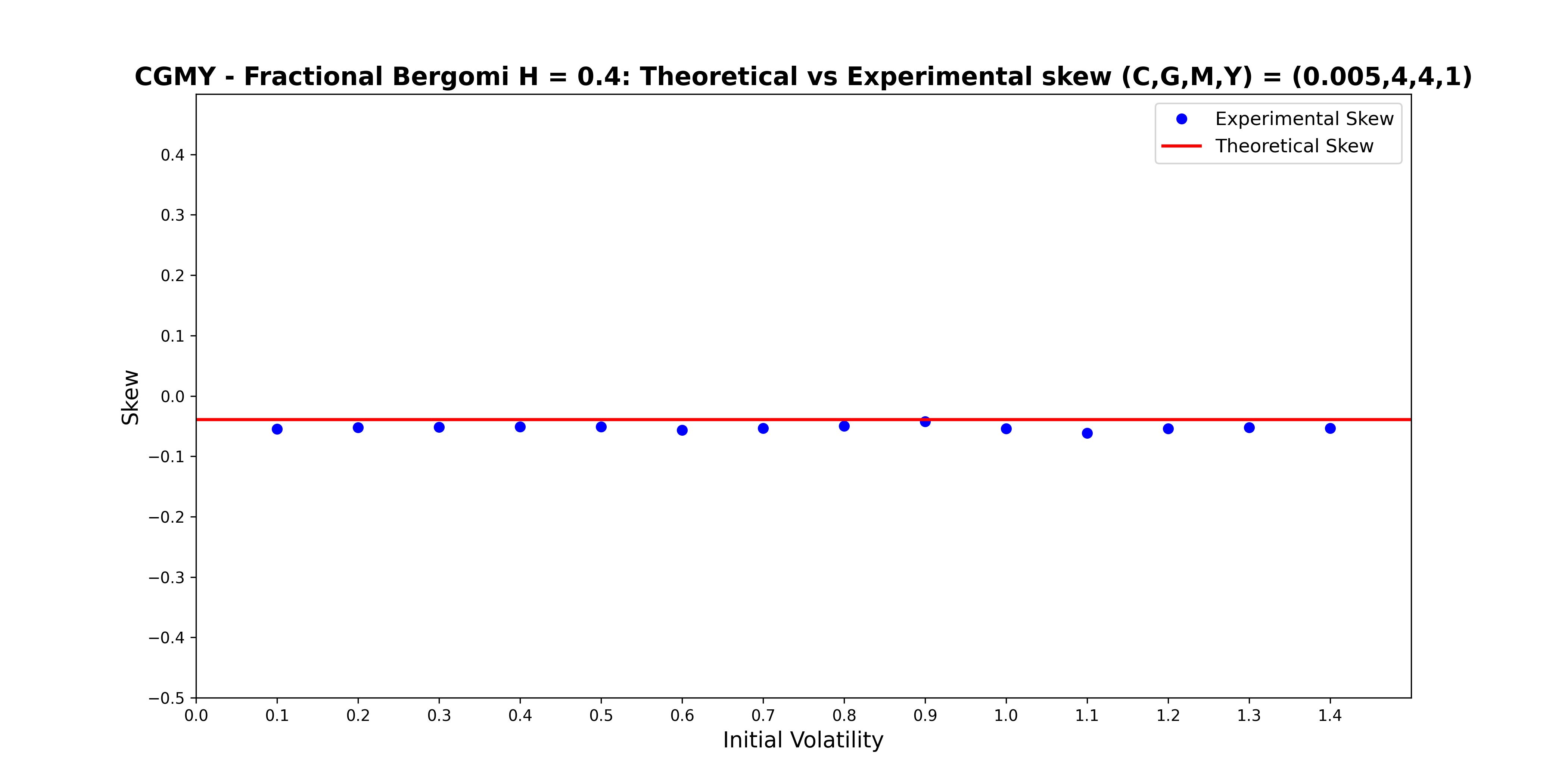}
        \caption{In blue we find the simulated values $(\sigma_0, T^{-0.1}\partial_k I_0(\sigma_0))$ and we observe that they coincide with the theoretical curve $(\sigma_0,  \frac{2 \rho \alpha \sqrt{2H}}{3+4H(2+H)})$ plotted in red.}
        \label{fig: CGMY skew symmetric 0.4}
    \end{figure}
Notice that Figures \ref{fig: CGMY skew symmetric 0.7}, \ref{fig: CGMY skew symmetric 0.5} and \ref{fig: CGMY skew symmetric 0.4} coincide with the figures in \cite{alos2023implied} where no jumps were considered. Therefore, the objective of checking that Theorem \ref{th: main theorem} holds for an example of infinite activity with infinite variation trajectories Lévy process and checking that the results match with the model without jumps is successfully accomplished. 
\subsubsection{The asymmetric case}

Now we will choose $Y \in (1,2)$. In this case, the formulas for the ATM-IV level and skew according to Theorem \ref{th: main theorem} are the following.
\begin{ex}
    Let $S$ follow Equation \eqref{model of study} with $\sigma$ satisfying \eqref{fractional bergomi} Assume that $L$ is a pure-jump CGMY Lévy martingale. Then
    \[
    \lim_{T \to 0} I_0(k^*) = \sigma_0, \quad \lim_{T \to 0} T^{\max(1/2-H,0)} \partial_kI_0(k^*) = \begin{cases}
        \frac{C(M^{Y-1} - G^{Y-1})\Gamma(1-Y)}{ \sigma_0}, & H > 1/2, \\
        \frac{C(M^{Y-1} - G^{Y-1})\Gamma(1-Y)}{\sigma_0} + \frac{\rho\alpha}{4}, & H = 1/2, \\
        \frac{2 \rho \alpha \sqrt{2H}}{3+4H(2+H)} & H < 1/2.
    \end{cases}
    \]
In order to numerically validate the previous identities We choose, for instance, $C = 0.05$, $G = 2, M = 4, Y = 1.5$ so that
\[
c_1 = \lim_{\eps \to 0} c_1^{\eps} = C(M^{Y-1} - G^{Y-1})\Gamma(1-Y)  =: c_1\approx -0.10382794271800314.
\]
For the ATM-IV level we assume that $\sigma$ follows the dynamics given by \eqref{fractional bergomi} with parameters
\[
T = 0.001, \quad S_0 = 100, \quad \alpha = 0.5, \quad \rho = -0.3, \quad H = 0.4,
\]
 In Figure \ref{fig: CGMY level asymmetric 0.4} we see that the experimental results are aligned with the conclusion of Theorem \ref{th: main theorem}.
\begin{figure}[H]
        \centering
        \includegraphics[width = 13cm, height = 6cm]{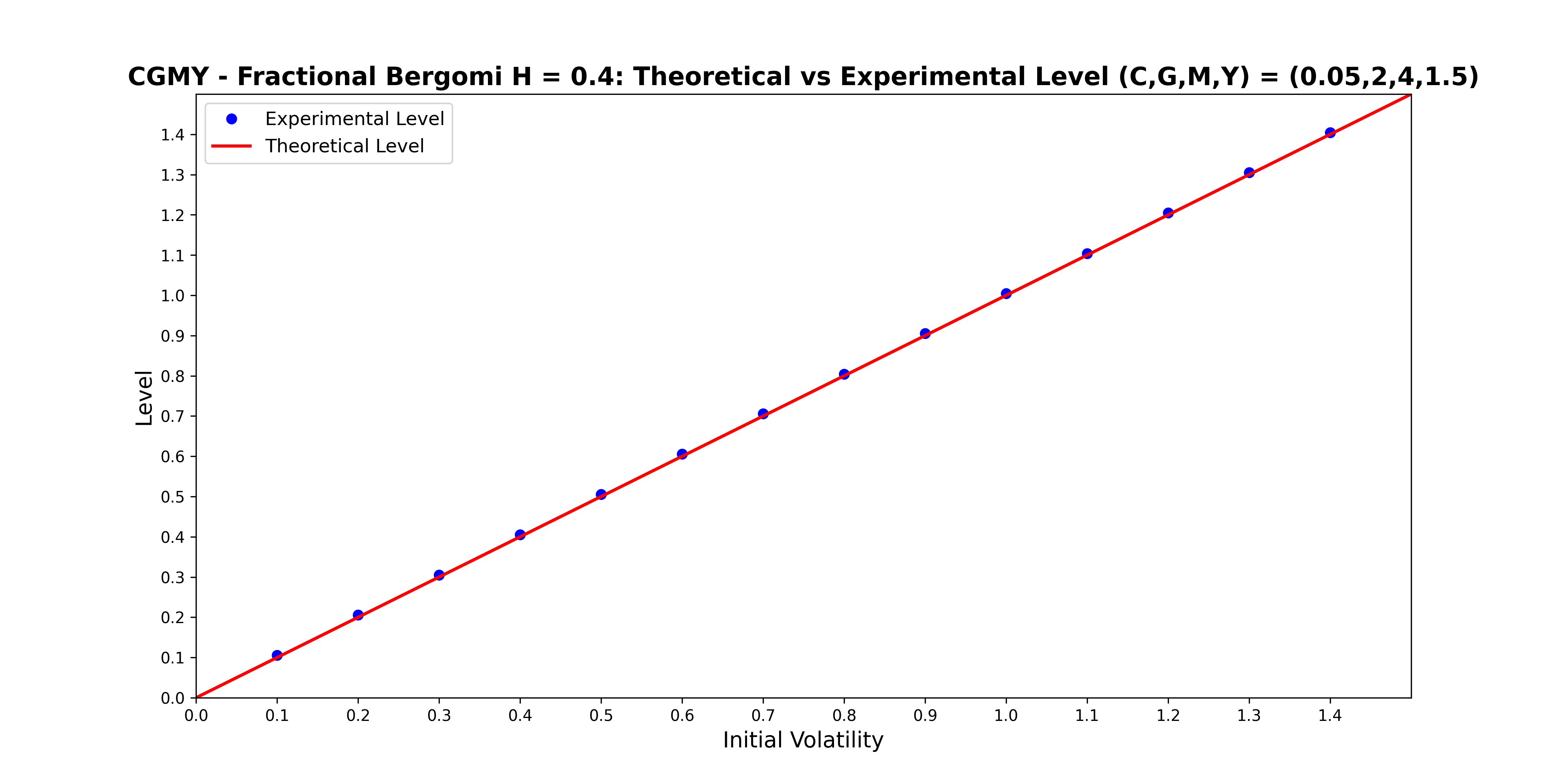}
        \caption{In blue we find how the simulated pairs $(\sigma_0, I_0(k^*))$ fit in the theoretical curve $(\sigma_0, \sigma_0)$. }
        \label{fig: CGMY level asymmetric 0.4}
    \end{figure}
Regarding the experiment with the skew, we select the set of parameters 
\[
T = 0.001, \quad S_0 = 100, \quad \alpha = 0.5, \quad \rho = -0.3, \quad H = 0.7.
\]
In Figure \ref{fig: CGMY skew asymmetric 0.7} we can observe how the simulated ATM-IV skew matches the theoretical skew.
  \begin{figure}[H]
        \centering
        \includegraphics[width = 13cm, height = 6cm]{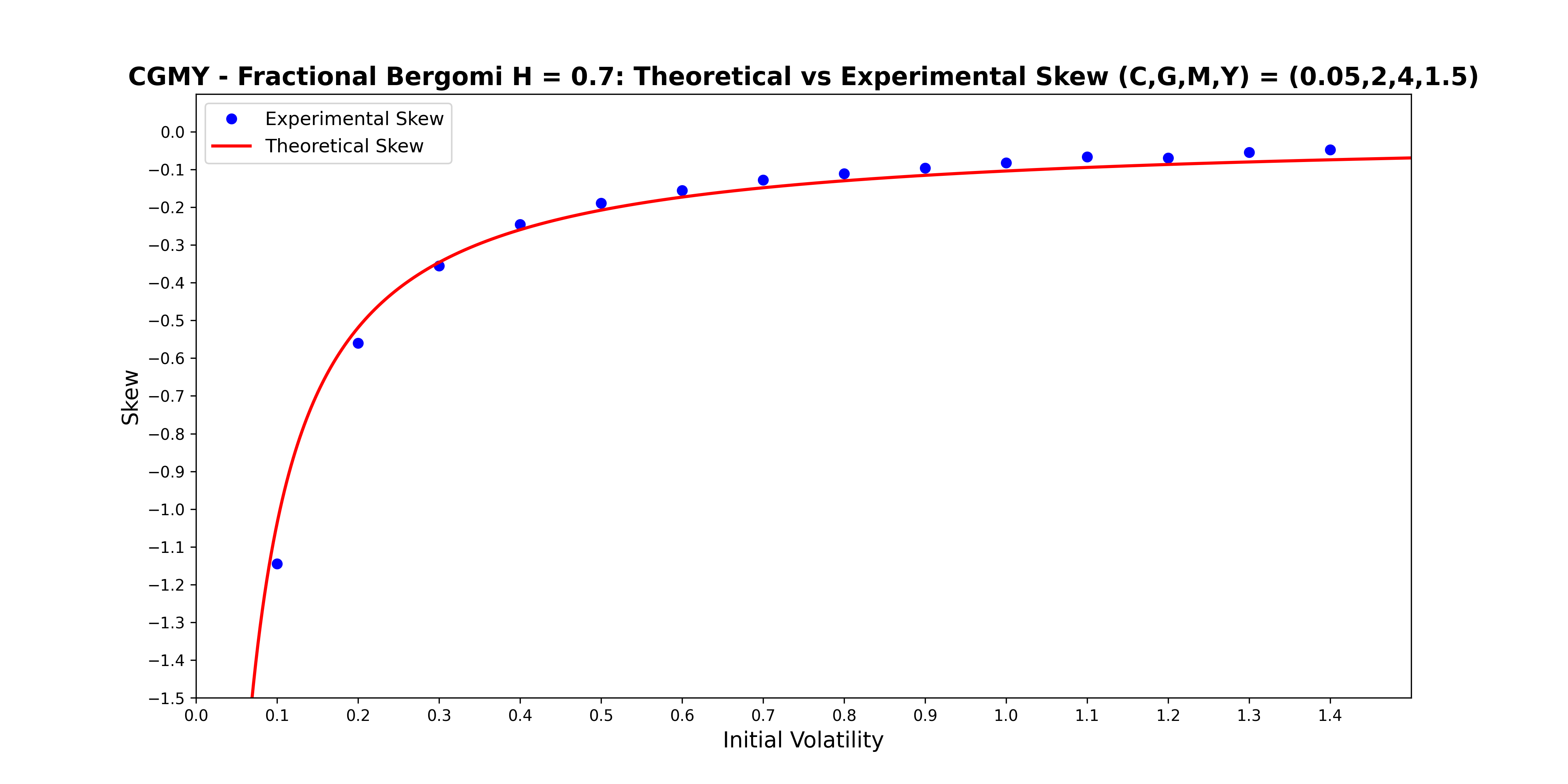}
        \caption{In blue we find the simulated pairs $(\sigma_0, \partial_k I_0(k^*))$ against the theoretical curve $(\sigma_0, \frac{c_1}{\sigma_0})$ plotted in red.}
        \label{fig: CGMY skew asymmetric 0.7}
    \end{figure}
\end{ex}
\subsection{The Normal Inverse Gaussian case}
Assume now that the jumps of the process $L$ follow a  Normal Inverse Gaussian (NIG) distribution (see \cite{barndorff1997processes}, \cite{aguilar2021explicit}). In this case, $L$ is a Lévy process with infinite activity and infinite variation with Lévy measure
\begin{align*}
    &\nu(\d y) = f(y; \alpha, \beta, \delta) \d y \\
    &f(y; \alpha, \beta, \delta) = \frac{ \delta \alpha }{\pi |y|}\exp(\beta y) K_1(\alpha |y|),
\end{align*}
where $K_1$ denotes the modified Bessel function of second kind and index 1, also known as the Macdonald function (for a complete reference on this function, check the reference \cite{spanier1987atlas}). In this particular case, we can check that the hypotheses of Theorem \ref{th: main theorem} are not satisfied for the skew. Indeed, one can check $\lim_{\eps \to 0} c_1^{\eps} = \infty$ because $K_1(\alpha |x|) \sim \frac{1}{\alpha |x|}$ for $|x| \to 0$ and $K_1(\alpha|x|) \sim |x|^{-1/2} e^{-\alpha|x|}$ when $|x| \to \infty$.  We expect then some phenomena that makes the implied volatility surface non-differentiable ATM. In order to make a plot of this phenomena, we assume that $\sigma_t$ follows the same dynamics as the SABR stochastic volatility model, i.e.
\[
\sigma_t  = \sigma_0 \exp\left( \hat{\alpha} W_t - \frac{\hat{\alpha}^2 t}{2}\right)
\]
with $\sigma_0 = 0.2$ and $\hat{\alpha} = 0.5$.  We choose again $\rho = -0.3$. The parameters for the Lévy process that we choose  are
\[
\alpha = 1.5, \quad \delta  = 1.0, \quad \beta = 0.5.
\]
As we can see in Figure \ref{fig: Implied volatility NIG}, there is a blow-up in the volatility surface coordinate curve of fixed time $T = 10^{-5}$ and variable strike. However, it is not direct from the picture but it can be observed that the conclusion regarding the level holds, i.e. 
\[
\lim_{T \to 0} I_0(k^*) = \sigma_0.
\]
In order to make it clearer, we can check in Table \ref{tab: implied volatility NIG} the values of the implied volatility for different strikes, $S_0 = 100$ and $T = 10^{-5}$.
\begin{table}[h!] \label{implied volatility NIG table}
\centering
\begin{tabular}{|c|c|}
\hline
\textbf{$k$} & \textbf{$I_0(k)$} \\ \hline
$95 $      &$ 376.9346295103556$ \\ \hline
$96 $      & $304.9319294669602$ \\ \hline
$97   $    & $232.03812003788872 $\\ \hline
$98 $      & $158.01739738115333$ \\ \hline
$99$       & $82.0954389422525$ \\ \hline
$100  $    & $0.198425367071128 $\\ \hline
$100.001 $ & $0.19913824561039337$ \\ \hline
$100.002 $ & $0.19929381260883816$ \\ \hline
$100.003 $ & $0.18570061448085348$ \\ \hline
$100.004 $ & $0$ \\ \hline
\end{tabular}
\caption{Implied volatility as a function of the strike when $T = 10^{-5}$.}
\label{tab: implied volatility NIG}
\end{table}

\begin{figure}[H]
        \centering
        \includegraphics[width = 13cm, height = 6cm]{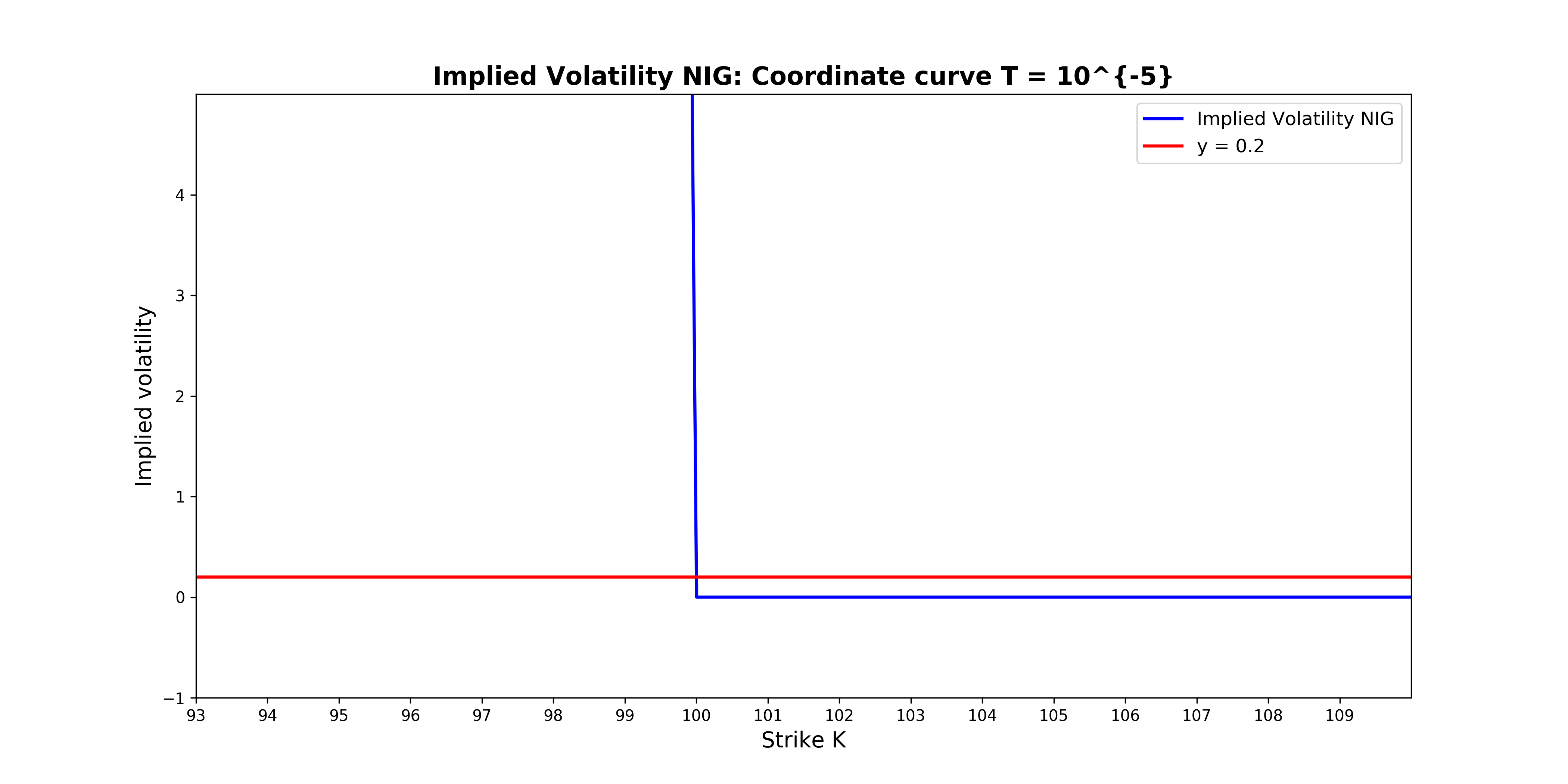}
        \caption{ATM-IV coordinate curve $I_0(k)$ for $k$ in a neighborhood of $k^*= 100$ and fixed maturity $T = 10^{ -5}$.}
        \label{fig: Implied volatility NIG}
    \end{figure}

\subsection{Summary of the numerical experiments}

The conclusion we can extract from these numerical experiments is that, when we fix the stochastic volatility model, the level remains the same for all Lévy processes that we add to the model regardless of its activity and the variation of the paths. We can see in Figures \ref{fig: Bates model different means}, \ref{fig: level and skew bates laplace & gaussian} and \ref{fig: Bates model Gaussian jumps zoom} the $\frac{c_1}{\sigma_0}$ additive factor in the skew that appears in Theorem \ref{th: main theorem}. Moreover, Figures  \ref{fig: CGMY skew symmetric 0.7}, \ref{fig: CGMY skew symmetric 0.5} and \ref{fig: CGMY skew symmetric 0.4} reflect that in the infinite activity and infinite variation case, the skew behaves in the same way as if there were no jumps (see \cite{alos2023implied}) due to the fact that $c_1 = \lim_{\eps \to 0}c_1^{\eps} = 0$. Moreover, Figures \ref{fig: CGMY skew asymmetric 0.7} and \ref{fig: Implied volatility NIG} show, respectively, the effect of $c_1 = \lim_{\eps \to 0}c_1^{\eps}$ in the case where $c_1$ is a real number and the non-differentiability of the volatility surface at the ATM point that occurs when $c_1^{\eps} \to \infty$.

In order to summarize the different behaviors of the skew, We can see in Figure \ref{fig: skew comparison 0.4} how the jumps don't affect the order of blow-up. Instead, a variation on the first order moment of $\nu$ causes a shift in the ATM-IV skew as a function of the time to maturity $T$. In other words, Figure \ref{fig: skew comparison 0.4} shows that the effect of the jumps contributes additively to the skew, not affecting the speed of the blow-up.
\begin{figure}[H]
        \centering
        \includegraphics[width = 12cm, height = 5cm]{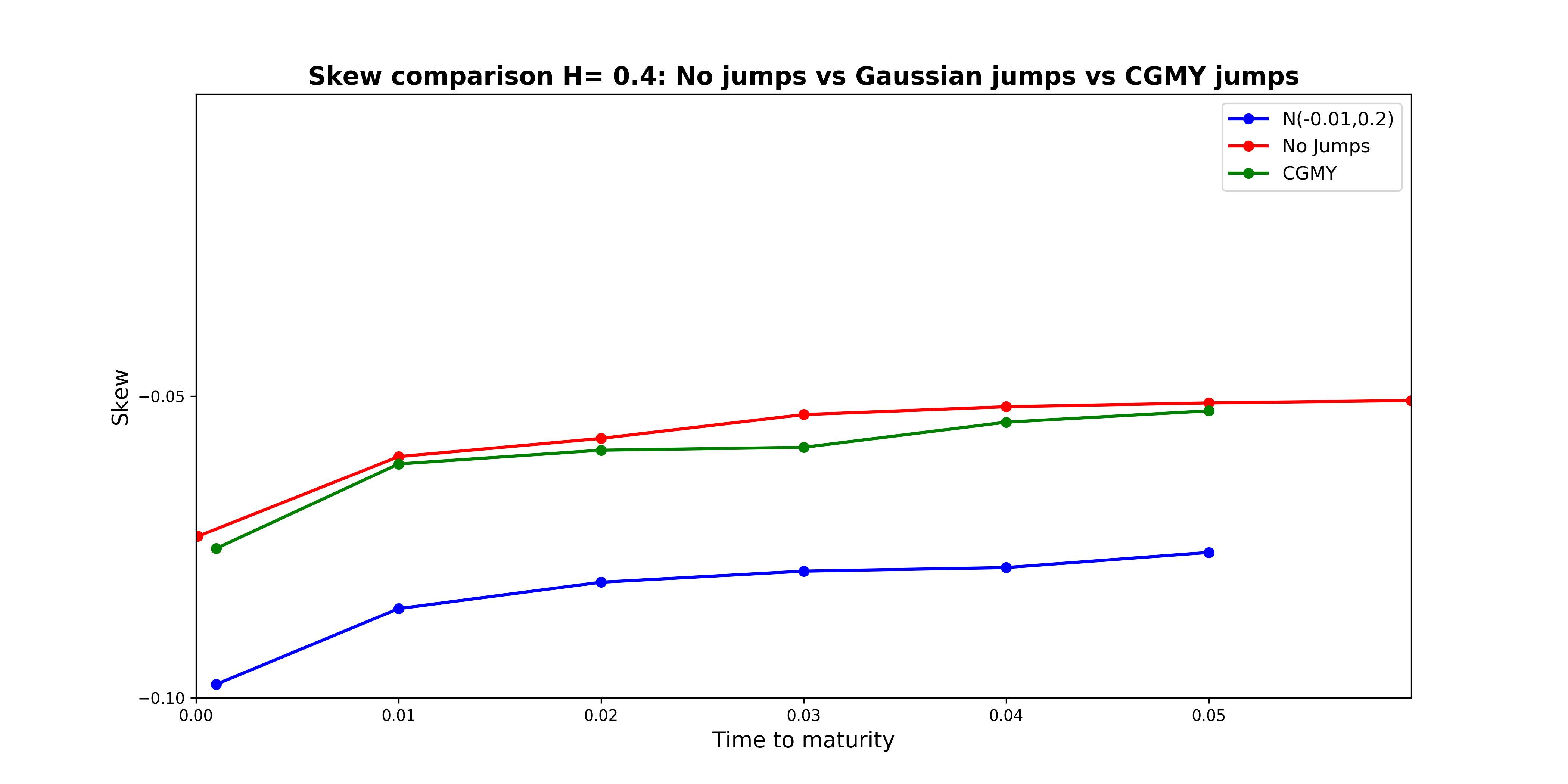}
        \caption{ATM-IV skew as a function of the time to maturity in the cases where the jumps are modeled with a compound Poisson process with Gaussian jumps (blue), in the case where the jumps follow a CGMY distribution (green) and in the case we don't consider jumps (red).}
        \label{fig: skew comparison 0.4}
    \end{figure}

Moreover, in Figure \ref{fig: skew comparison 0.7} we can see how the skew behaves (as a function of $\sigma_0$) when $H = 0.7$ and the jumps follow either a Laplace distribution, a CGMY distribution or there are no jumps.

\begin{figure}[H]
        \centering
        \includegraphics[width = 13cm, height = 6cm]{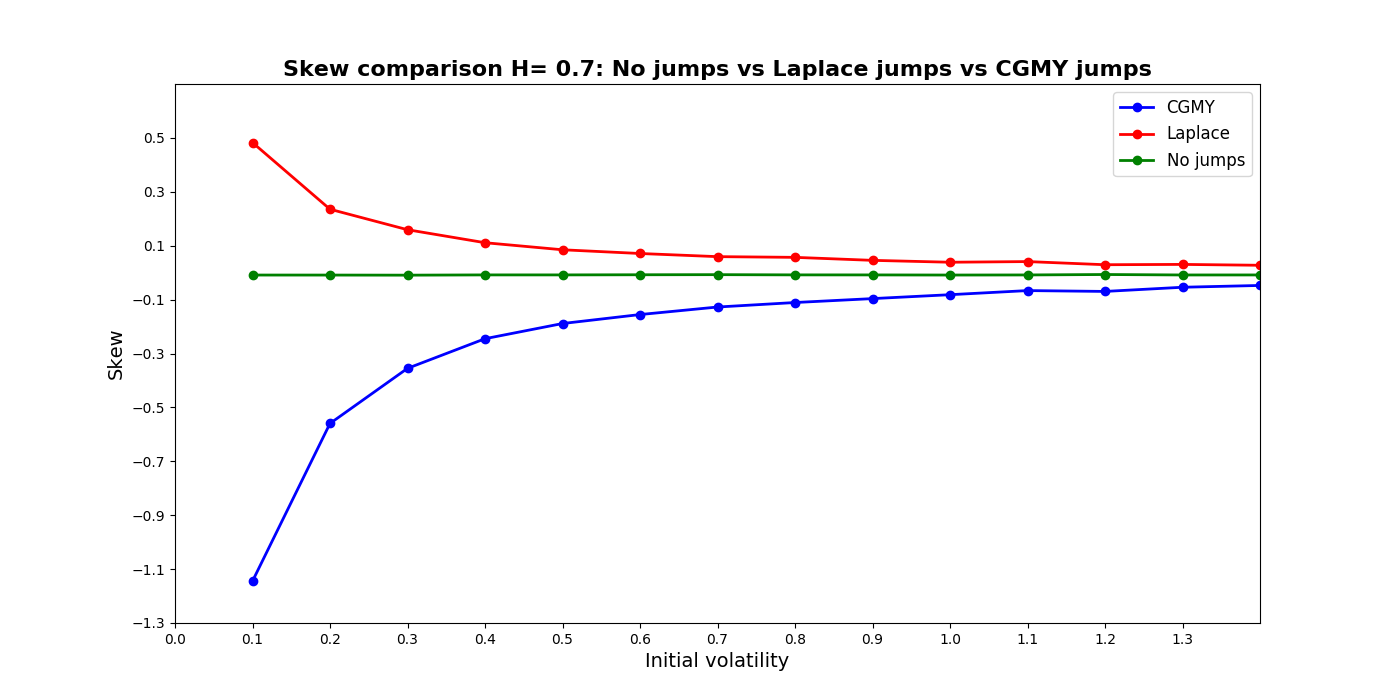}
        \caption{ATM-IV skew as a function of $\sigma_0$ for two different Lévy process and when no jumps are considered.}
        \label{fig: skew comparison 0.7}
    \end{figure}

\section{Conclusions}
In this paper we have seen how the Malliavin calculus provides analytic formulas for the short-time behavior for the ATM-IV level and skew for the jump-diffusion stochastic volatility Bachelier model in the case where the Lévy process has finite activity and finite variation. Moreover, since any Lévy process can be approximated by a compound Poisson process, we have shown how a direct approximation argument extends the results for a large class of Lévy processes. The numerical experiments confirm that the presence of the Lévy process does not change the short-time behavior of the ATM-IV level but it has an effect on the ATM-IV skew, acting in an additive way and depending only on the mean of the jumps.

\bibliographystyle{apalike}
\bibliography{references.bib}

\end{document}